\documentclass[11pt]{article}

\usepackage[letterpaper,margin=1in]{geometry}
\usepackage[T1]{fontenc}
\usepackage{amsmath,amssymb,amsthm,mathtools}
\usepackage{newtxtext,newtxmath}
\usepackage{thmtools}
\usepackage{graphicx}
\usepackage{float}
\usepackage{microtype}
\usepackage{comment}
\usepackage{xcolor}
\usepackage{url}
\usepackage[round,authoryear]{natbib}
\usepackage{hyperref}

\definecolor{linkblue}{RGB}{22,68,122}
\definecolor{citegreen}{RGB}{23,105,74}
\hypersetup{
  colorlinks=true,
  linkcolor=linkblue,
  citecolor=citegreen,
  urlcolor=linkblue,
  pdfauthor={Anonymous Authors},
  pdftitle={A Robustified Greedy Algorithm for Online Transportation with Improved Competitive Guarantees}
}

\newtheorem{theorem}{Theorem}
\newtheorem{lemma}[theorem]{Lemma}
\newtheorem{corollary}[theorem]{Corollary}

\newcommand{\OPT}{\mathrm{OPT}}

\newcommand{\ALG}{\mathrm{ALG}}

\newcommand{\diam}{\operatorname{DIAM}}
\newcommand{\tsp}{\operatorname{TSP}}
\newcommand{\opt}{\operatorname{opt}}

\title{A Robustified Greedy Algorithm for Online Transportation\\
with Improved Competitive Guarantees}
\author{
Ritesh Seth\\
IIIT Delhi
\and
Syamantak Das\\
IIIT Delhi
\and
Sharath Raghvendra\\
North Carolina State University
}
\date{}

\begin{document}
\maketitle

\begin{abstract}
We study the \emph{online transportation problem}, in which $n$ requests arriving sequentially in a metric space must be irrevocably assigned to $k$ capacitated facilities. Beyond classical logistics applications, this problem models resource-allocation tasks arising in machine learning, including online facility assignments, recommender systems, and mixture-of-experts routing.

We introduce \emph{Robustified Greedy} (RG), a deterministic generalization of the Robust Matching algorithm that achieves a competitive ratio of $6.6604k-2.89$, improving upon the state-of-the-art bounds of $8k-7$ (Arndt et al., SOSA 2026) and $8k-5$ (Harada and Itoh, ICALP 2025). RG also retains the metric-sensitive guarantee established for Robust Matching (RM) (Nayyar and Raghvendra, FOCS 2017), achieving a competitive ratio of $O(k^{1-1/d}\log^2 n)$ in $d$-dimensional Euclidean spaces for fixed $d>1$. No comparable metric-sensitive guarantee is known for the transportation algorithms of Arndt et al.\ or Harada and Itoh.

Beyond these competitive guarantees, RG provides a simple explanation for its decisions. It favors the natural nearest-neighbor assignment and, for suitable parameters, departs from this choice only when it identifies a reassignment that reduces the cost of its maintained auxiliary matching, thereby correcting accumulated assignment costs. We also prove that nearest-neighbor assignments account for a guaranteed fraction of RG's total cost, approaching one-half for appropriate parameters, even under adversarial arrivals. Experiments on real-world datasets corroborate the theory: RG achieves lower cost-to-\textsc{Opt} ratios than the competing algorithms while retaining a substantial nearest-neighbor component in its cost.
\end{abstract}

\section{Introduction}

In the \emph{online metric transportation problem}, we are given a set $S$ of $k$ server locations in a metric space. Each $s\in S$ has a positive integral capacity $b_s$, with $\sum_{s\in S}b_s=n$. A sequence of $n$ requests arrives online, each of which must be irrevocably assigned to a server with available capacity, consuming one unit of capacity. The objective is to minimize the total assignment distance. When every server has unit capacity, this is the \emph{online metric matching problem}~\citep{kalyanasundaram1993online}. These problems model sequential allocation in logistics and resource assignment~\citep{harada2025nearly,arndt2026competitive}, as well as machine-learning settings in which requests and resources are embeddings, costs reflect embedding-space dissimilarity, and resources have limited capacities. Examples include capacity-constrained recommendation~\citep{mladenov2020social,li2024twosided}, balanced token-to-expert routing in sparse mixture-of-experts models~\citep{lewis2021base,zhou2022expert}, and $1$-Wasserstein-based imitation learning in robotics~\citep{dadashi2021pwil}.

Assignments must be made before future requests are known but are judged against an optimal solution that knows the entire sequence. The standard measure is the competitive ratio~\citep{sleator1985amortized}: an algorithm is $c$-competitive if, for every request sequence and arrival order, its cost is at most $c$ times the optimum. For transportation, the goal is a guarantee depending on the number of server locations $k$, independently of the total capacity $n$.

The natural Greedy algorithm assigns each request to a nearest available server. This minimizes the current cost and makes the decision easy to explain; in emergency response, for example, proximity to an available responder may be an immediate priority. Greedy also performs well empirically: \citet{tong2016online} found that it outperformed all other evaluated algorithms on almost all tests, including algorithms with substantially better worst-case guarantees. Under independent uniform sampling of servers and requests in $[0,1]^d$, Greedy is $O(1)$-competitive for every fixed $d\neq2$ and $\Theta(\sqrt{\log n})$-competitive for $d=2$~\citep{balkanski2023power,yang2026optimal}.

Nevertheless, Greedy has a worst-case competitive ratio of $2^n-1$. To see the lower bound without relying on tie-breaking, fix an arbitrarily small $\varepsilon>0$ and consider unit-capacity servers on the line at
$$
-(1+\varepsilon),\ 1,\ 3,\ 7,\ \ldots,\ 2^{n-1}-1,
$$
with requests arriving at
$$
0,\ 1,\ 3,\ 7,\ \ldots,\ 2^{n-1}-1.
$$
For request $0$, server $1$ is uniquely nearest. Subsequently, whenever a request arrives at an occupied server, the next server to its right is closer than the remaining server at $-1-\varepsilon$. Greedy therefore assigns request $0$ to server $1$, request $1$ to server $3$, and so on, leaving the final request to be assigned to server $-1-\varepsilon$. Its costs are $
1,\ 2,\ 4,\ \ldots,\ 2^{n-2},\ 2^{n-1}+\varepsilon$,
adding up to $2^n-1+\varepsilon$. The optimum instead assigns the first request to the server at $-(1+\varepsilon)$ and every other request to its own location, at total cost $1+\varepsilon$. Hence the competitive ratio is
$
\left(\frac{2^n-1+\varepsilon}{1+\varepsilon}\right)
$
which approaches $2^n-1$ as $\varepsilon\to0$. Each greedy decision is locally optimal but consumes capacity that the next request could have used at zero cost, causing a cascade of successively longer assignments. Thus, a rule that performs well on many inputs can be exploited by an adversarial sequence.

A retrospective alternative to Greedy maintains an optimal offline matching of the requests seen so far and uses it to guide the next irrevocable assignment. For online metric matching, the Permutation algorithm~\citep{kalyanasundaram1993online} follows this approach and achieves the optimal deterministic competitive ratio of $2n-1$. The RM algorithm~\citep{raghvendra:APPROX-RANDOM.2016:RM} instead combines immediate assignment costs with retrospective information. Using a parameter $t>1$, it balances the cost of serving the current request against the quality of a maintained auxiliary matching. Its analysis also adapts to the underlying metric, yielding competitive ratios of $O(\sqrt{n}\log^2 n)$ in two-dimensional Euclidean space and $O(n^{1-1/d}\log^2 n)$ in metrics of doubling dimension $d>1$~\citep{Nayyar2017Mu}.

These guarantees for online matching do not directly provide the desired guarantees for online transportation, where the number of server locations $k$ can be much smaller than the total capacity $n$. Representing each unit of capacity as a separate server allows matching algorithms to be applied, but their general-metric guarantees may still depend on $n$.~\cite{KalyanasundaramPruhs1998}  conjectured that online transportation
admits a deterministic $(2k-1)$-competitive algorithm, matching the known lower bound.  Obtaining a guarantee linear in $k$, independently of the total capacity, remained open for more than two decades.
Recent work has made substantial progress toward this conjecture. \citet{harada2025nearly} gave a deterministic $8k-5$-competitive algorithm for online transportation and \citet{arndt2026competitive} simplified this algorithm and improved the analysis to $8k-7$. 

The difficulty of extending guarantees on matching algorithms to the transportation setting is particularly clear for the Permutation algorithm. Despite being worst-case optimal for online metric matching, its competitive ratio can be as large as $2n-1$ for the online transportation problem on the line with only four distinct server locations~\citep{nicoletti2026permutation}. For RM, \citet{arndt2026competitive} observe an $O(k\log^2 n)$ bound and identify an $O(k)$ guarantee as an open analysis question. RM also retains metric-sensitive guarantees, whereas comparable guarantees are not known for the transportation algorithms of Harada and Itoh or Arndt et al.

Although these algorithms provide strong competitive guarantees, none of the online matching or transportation algorithms are known to guarantee that a fixed positive fraction of their total online cost comes from nearest-neighbor greedy assignments. 
This motivates the following question:
\begin{quote}
\emph{Can we design an online metric transportation algorithm
that favors greedy assignments while providing strong
worst-case guarantees?}
\end{quote}

\noindent
\textbf{Our Contribution.}
We introduce \emph{Robustified Greedy} (RG), a parameterized generalization of RM. As in RM, the algorithm maintains an irrevocable online matching and an auxiliary offline matching. RG uses a second parameter $\alpha \geq 1$. Setting $\alpha=t$ recovers RM, whereas $\alpha<t$ favors assignments to nearest available servers.

\begin{theorem}\label{thm:main}
The competitive ratio of RG is upper bounded by 
$6.6604k - 2.89$.
\end{theorem}

This guarantee is independent of the total capacity $n$ and improves upon the previous best-known bound of $8k-7$ by~\cite{arndt2026competitive}. We first provide a somewhat simpler analysis that yields a competitive ratio of $(20/3)k - 3$, while the tighter analysis appears in Section~\ref{appdx:parametric}.

Using our analysis framework, we could also prove a competitive ratio of $11.66k$ for the RM algorithm ($\alpha = t$ regime for RG), thereby resolving a question raised in~\cite{arndt2026competitive} (see Section~\ref{appdx:rm-analysis}). However, we give a much more nuanced analysis for RG, utilizing the greedy choices that the algorithm explicitly makes. Favoring such greedy assignments creates persistent structure in the auxiliary matching - every positive-cost greedy edge remains unchanged until its server becomes saturated. We track this structure through \emph{server saturation epochs} to obtain the improved competitive bound of $6.6604k$ for $t = 2 + \sqrt{5}, \alpha = \sqrt{t}$.   

We also recover the metric-sensitive guarantees of RM for the RG algorithm.
\begin{theorem}[Metric-sensitive guarantee]
\label{thm:metric-sensitive}
For every fixed $t>1$ and $1\le\alpha\le t$, RG has competitive ratio
$
O\!\left(k^{1-1/d}\log^2 n\right)$
in $d$-dimensional Euclidean spaces, for every fixed $d>1$. 
%
\end{theorem}

The parameter $\alpha$ biases the RG algorithm towards making more nearest neighbor choices. Indeed, we show that with suitable choice of parameters, a significant fraction of the online cost of RG comes from nearest neighbor assignments, while it still retains a guarantee of $O(k)$ for the competitive ratio. However, greedy assignments increase the cost of the maintained auxiliary offline matching and the accumulated cost can potentially start affecting the online cost.  
The algorithm makes a non-greedy assignment \emph{only when it finds evidence that this accumulated cost can be reduced}.

\begin{theorem}
\label{thm:combined}
Suppose $t \geq 3 + 2\sqrt{2}$ and $1\leq \alpha \leq \sqrt{t}$ are fixed constants. Then  
(a) RG algorithm is $O(k)$-competitive, and,
(b) The total fraction of online cost of RG from greedy assignments is at least 
$\left(
\frac12-\frac{\alpha+1}{2(t-\alpha)}
\right).
$
\end{theorem}

  Note that for fixed $\alpha=o(t)$, the guaranteed fraction in Theorem~\ref{thm:combined}(b) approaches 1/2 as $t$ increases. Complete proofs of Theorems~\ref{thm:metric-sensitive} and~\ref{thm:combined} appear in Sections~\ref{appdx:metric-sensitive} and~\ref{appdx:parametric}, respectively.  

\noindent
\textbf{Experiments.} Experiments on real-world datasets show that RG retains Greedy's strong performance on natural arrival orders while being more robust to correlated arrivals, improving over Greedy by up to $15.02\%$. RG also consistently outperforms the algorithm due to ~\cite{arndt2026competitive} while preserving a high fraction of nearest-neighbor assignments.

\paragraph{Preliminaries.}
Let $(X,d)$ be a metric space, where $d(u,v)$ denotes the distance between any two points $u,v\in X$. For a set of edges $E$, define its \emph{cost} by $
w(E) := \sum_{(u,v)\in E} d(u,v)$.

Recall that there are $k$ servers in $S$, each with a positive integral capacity, and that the total capacity of the servers is $n$. Let $R$ denote a set of $n$ requests. A \emph{matching} $M$ assigns requests in $R$ to servers in $S$ such that each request is incident to at most one edge of $M$, and the number of edges in $M$ incident to any server does not exceed its capacity. A request is \emph{free} with respect to $M$ if it is unmatched, while a server is \emph{free} if it has unused capacity. A server whose capacity is fully used is said to be \emph{saturated}. The \emph{size} of $M$, denoted by $|M|$, is the number of requests matched by $M$. Let $\OPT$ denote the cost of a minimum-cost matching of all $n$ requests to the servers, subject to their capacities.

An \emph{alternating path} with respect to a matching $M$ is a path whose edges alternate between edges in $M$ and edges not in $M$. An \emph{augmenting path} is a simple alternating path that starts at a free request and ends at a free server. Given an augmenting path $P$, augmenting $M$ along $P$ removes the edges $M\cap P$ from $M$ and adds the edges $P\setminus M$. We denote the resulting matching by $
M \leftarrow M\oplus P$.
Augmenting along an augmenting path increases the size of the matching by exactly one.

Fix parameters $t>1$ and $1\le\alpha\le t$. An augmenting path consisting of a single edge between a free request and a free server is called a \emph{direct augmenting path}. All other augmenting paths contain at least three edges and are called \emph{non-direct augmenting paths}. 

The $(t,\alpha)$-net-cost of an augmenting path $P$ with respect to $M^*$ is
\begin{equation}\label{eq:netcost}
\phi_{t,\alpha}^{M^*}(P)=
\begin{cases}
\alpha d(r,s), & P=\{(r,s)\}\text{ is direct},\\[2pt]
t\,w(P\setminus M^*)-w(P\cap M^*), & P\text{ is non-direct}.
\end{cases}
\end{equation}
For simplicity, we omit the superscript $M^*$ when the matching is clear from the context and write $\phi_{t,\alpha}(P)$. When $\alpha=t$, this definition coincides with the $t$-net-cost used in the Robust Matching algorithm from~\cite{Nayyar2017Mu}. 





\section{Robustified Greedy Algorithm}
\label{sec:robustified-greedy}
Fix parameters $t>1$ and $1\le\alpha\le t$. Our algorithm maintains an online matching $M$ and an offline matching $M^*$, both initially empty. After each arrival, both matchings serve all requests received so far and assign the same number of requests to each server. To process a new request $r$, the algorithm performs the following steps:
\begin{enumerate}
\item Compute an augmenting path $P$ from $r$ to a free server $s$ that minimizes the $(t,\alpha)$-net-cost with respect to $M^*$.
\item Augment the offline matching along $P$, setting
$M^*\leftarrow M^*\oplus P$.
\item Assign $r$ to $s$ in the online matching, setting
$M\leftarrow M\cup\{(r,s)\}$.
\end{enumerate}
These updates ensure that \(M\) and \(M^*\) continue to assign the same number of requests to each server. Augmenting \(M^*\) increases the number assigned to the terminal server \(s\) by one and leaves this number unchanged at every other server. The online assignment makes the same change in \(M\). Thus, the two matchings have identical remaining capacities and the same free servers after every arrival.

For a request \(r\), the \((t,\alpha)\)-net-cost of a direct augmenting path to a free server \(s\) is \(\alpha d(r,s)\). This cost is minimized by a nearest free server. Therefore, whenever the algorithm selects a direct augmenting path, it makes a greedy assignment.

When \(\alpha=t\), the \((t,\alpha)\)-net-cost coincides with the \(t\)-net-cost defined in~\cite{raghvendra:APPROX-RANDOM.2016:RM}, so our algorithm recovers the Robust Matching algorithm. Choosing \(\alpha<t\) discounts direct augmenting paths while leaving the cost of non-direct paths unchanged, thereby favoring nearest-server assignments.

The following lemma relates the online matching cost to the total $(t,\alpha)$-net-cost of the selected paths and the contribution of greedy choices. It is central both to explaining the algorithm’s decisions and to analyzing its competitive ratio.


Consider any consecutive sequence of arrivals. Let
$M_{\mathrm{start}}^*$ and $M_{\mathrm{end}}^*$ be the offline
matchings immediately before and after this sequence. Let $G$ be
the total weight of the selected greedy edges, and let $\Phi$ be
the sum of the $(t,\alpha)$-net-costs of all selected paths.

\begin{restatable}[Cost accounting]{lemma}{costAccountingLemma}
\label{lem:accounting}
The online cost incurred during the sequence is at most
\begin{equation}
\label{eq:online}
\frac{
    2\Phi+2(t-\alpha)G
    +(t+1)\bigl(
        w(M_{\mathrm{start}}^*)-w(M_{\mathrm{end}}^*)
    \bigr)
}{t-1}.
\end{equation}
\end{restatable}

\begin{proof}
For each selected non-direct path $P$, let
\[
I=w(P\setminus M^*),\qquad D=w(P\cap M^*),
\]
where $M^*$ is the matching immediately before augmentation. Since
$w(P)=I+D$, $\phi_{t,\alpha}(P)=tI-D$, and augmentation changes
$w(M^*)$ by $I-D$,
\[
2\phi_{t,\alpha}(P)
=2(tI-D)
=(t-1)w(P)+(t+1)(I-D)
\]

Let $L$ be the total weight of the selected non-direct paths. The
direct paths contribute $\alpha G$ to $\Phi$ and increase $w(M^*)$
by $G$; hence the total change due to the non-direct paths is
$w(M_{\mathrm{end}}^*)-w(M_{\mathrm{start}}^*)-G$. Summing the
preceding identity gives
\[
2\Phi
=2\alpha G+(t-1)L
 +(t+1)\bigl(w(M_{\mathrm{end}}^*)-w(M_{\mathrm{start}}^*)-G\bigr).
\]
Therefore,
\[
(t-1)(G+L)
=2\Phi+2(t-\alpha)G
 +(t+1)\bigl(w(M_{\mathrm{start}}^*)-w(M_{\mathrm{end}}^*)\bigr).
\]
Each online edge joins the endpoints of its selected path and hence,
by the triangle inequality, costs at most the path weight. The online
cost is therefore at most $G+L$, and the result follows.
\end{proof}
\subsection{Greedy Assignment Fraction}
\label{sec:explainability}

The following structural guarantees are independent of the competitive-ratio
analysis; their proofs are given below.

\begin{restatable}{lemma}{greedyOrReducingLemma}
\label{lem:reducing}
Suppose $1\le\alpha<(t-1)/2$. On each arrival, the algorithm either
selects a greedy edge or selects a non-direct path $P$ satisfying
\[
w(M^*)-w(M^*\oplus P)
\ge
\frac{t-1-2\alpha}{t+1}\,w(P).
\]
\end{restatable}

\begin{proof}
Suppose that $P$ is non-direct, and write
\[
I=w(P\setminus M^*),\qquad D=w(P\cap M^*).
\]
If $d$ is the distance from the new request to a nearest free server,
then the corresponding direct edge has net-cost $\alpha d$. Since $P$
minimizes net-cost and its terminal server is free, the triangle
inequality gives
\[
tI-D\le\alpha d\le\alpha w(P)=\alpha(I+D).
\]
Using the identity from Lemma~\ref{lem:accounting},
\[
(t+1)(D-I)
=(t-1)(I+D)-2(tI-D)
\ge(t-1-2\alpha)(I+D).
\]
Therefore,
\[
D-I\ge
\frac{t-1-2\alpha}{t+1}(I+D).
\]
Augmentation decreases the offline matching cost by $D-I$, and the
coefficient is positive under the stated assumption.
\end{proof}

\begin{restatable}[Greedy-cost bound]{lemma}{greedyCostLemma}
\label{lem:greedyfraction}
Let $G$ be the total cost of selected greedy edges, and let $M$
and $M^*$ be the final online and offline matchings. For
$1\le\alpha<(t-1)/2$,
\[
G\ge
\left(
    \frac12-\frac{\alpha+1}{2(t-\alpha)}
\right)w(M).
\]
\end{restatable}

\begin{proof}
On each arrival, the direct edge to a nearest free server is a
candidate augmenting path. The selected path therefore has
net-cost at most $\alpha$ times the nearest-free-server distance.
This distance is at most the cost of the online assignment,
since the assignment also uses a free server. Summing over all
arrivals gives
$
\Phi\le\alpha w(M).
$
Applying Lemma~\ref{lem:accounting} to the entire execution,
whose initial offline matching is empty, yields
\[
(t-1)w(M)
\le
2\Phi+2(t-\alpha)G-(t+1)w(M^*).
\]
Substituting $\Phi\le\alpha w(M)$ and rearranging gives
\[
2(t-\alpha)G
\ge
(t-1-2\alpha)w(M)+(t+1)w(M^*).
\]
The lemma follows using $w(M^*)\ge0$. 
\end{proof}

Lemma~\ref{lem:reducing} quantifies the decrease in the offline matching cost
behind every non-greedy choice, while Lemma~\ref{lem:greedyfraction} shows that
the fraction of online cost contributed by greedy edges approaches $1/2$ for
fixed $\alpha$, or more generally $\alpha=o(t)$. This positive fraction relies
on discounting direct paths: when $\alpha=t$ (RM), the same argument gives no
positive lower bound on $G$.

\section{Primal-Dual Implementation}
\label{sec:primal-dual}
In this section, we present a primal-dual implementation of the Robustified-Greedy algorithm that processes each request in \(O(nk\log n)\) time. More importantly, the evolution of the dual weights maintained by the algorithm plays a critical role in our analysis of its competitive ratio (see Section~\ref{sec:competitive-ratio}).


\paragraph{Feasibility.} Recall that each server $s\in S$ has an integer capacity $b_s$.
The algorithm maintains a dual weight $y(v)$ for every request
and server $v$.

A matching $M^*$, together with the dual weights $y(\cdot)$, is
\emph{$(t,\alpha)$-feasible} if every request-server pair $(r,s)$
satisfies
\begin{equation}\label{eq:upper}
y(r)+y(s)\le
\begin{cases}
\alpha d(r,s),&\text{if both $r$ and $s$ are free},\\
t\,d(r,s),&\text{otherwise},
\end{cases}
\end{equation}
and every matching edge satisfies
\begin{equation}\label{eq:lower}
y(r)+y(s)\ge d(r,s)
\qquad\text{for every }(r,s)\in M^*.
\end{equation}

\paragraph{Slack.}
For each non-matching edge $(r,s)$, define its slack as
\begin{equation}\label{eq:forward}
\sigma(r,s)=
\begin{cases}
\alpha d(r,s)-y(r)-y(s),
    &\text{if both $r$ and $s$ are free},\\
t\,d(r,s)-y(r)-y(s),
    &\text{otherwise}.
\end{cases}
\end{equation}
For each matching edge $(r,s)$, define its slack as
\begin{equation}\label{eq:backward}
\sigma(s,r)=y(r)+y(s)-d(r,s).
\end{equation}
Feasibility guarantees that all these slacks are nonnegative.

\paragraph{Residual graph.}
Given a $(t,\alpha)$-feasible matching $M^*$ and its dual weights,
we construct the residual graph $G_{M^*}$ on the requests and
servers, together with an auxiliary sink vertex $z$. Each
non-matching edge $(r,s)$ is directed from $r$ to $s$ and assigned
weight $\sigma(r,s)$. Each matching edge $(r,s)$ is directed from
$s$ to $r$ and assigned weight $\sigma(s,r)$. Finally, we add a
directed edge of weight zero from every free server to $z$.




\subsection{Processing a request via primal-dual algorithm}
Initially, before any request has arrived the matching is empty and the dual weights for every server is set to zero. As requests arrive, our algorithm will update the dual weights to find the desired augmenting path. We describe this procedure for any request $r$. 

Before each arrival, the algorithm maintains a set of dual weights $y(\cdot)$ on all servers and requests and a
$(t,\alpha)$-feasible matching $M^*$
When a new request $r$ arrives, we initialize $y(r)=0$,
construct the residual graph $G_{M^*}$, and perform the
following steps.

\begin{enumerate}
\item \textbf{Find an augmenting path.}
Execute Dijkstra's shortest-path algorithm from $r$ in
$G_{M^*}$. Let $\ell(v)$ denote the shortest-path distance
from $r$ to $v$, and let $\delta=\ell(z)$.
Choose a shortest $r$-to-$z$ path and truncate it at its first free server $s$ to obtain $P$.

\item \textbf{Update the dual weights.}
For each request or server $v$ with $\ell(v)<\delta$, update
\begin{equation}\label{eq:update}
\begin{aligned}
y(v)&\leftarrow y(v)+\delta-\ell(v),
    &&\text{if }v\in R,\\
y(v)&\leftarrow y(v)-\delta+\ell(v),
    &&\text{if }v\in S.
\end{aligned}
\end{equation}
All other dual weights remain unchanged.

\item \textbf{Augment the matchings.}
Assign $r$ to $s$ in the online matching and augment the
offline matching along $P$:
\[
M\leftarrow M\cup\{(r,s)\},
\qquad
M^*\leftarrow M^*\oplus P.
\]
\end{enumerate}
Every server is initialized with dual weight zero. Each dual update
decreases the weight of a server $s$ by $\delta-\ell(s)$ when
$\ell(s)<\delta$ and leaves it unchanged otherwise. Thus, all server
dual weights remain nonpositive throughout the execution.

\paragraph{Efficiency.} The runtime for the RG algorithm per request arrival is dominated by the time to run Dijkstra's algorithm on a graph with $O(nk)$ edges, which is $O(nk\log n)$ (since $k \leq n$). 
\begin{restatable}{lemma}{runtimeLemma}\label{lem:runtime}
The Robustified Greedy algorithm processes each request
in $O(nk\log n)$ time.
\end{restatable}

\begin{proof}
The residual graph has $O(nk)$ edges. Executing Dijkstra's algorithm to find an
augmenting path takes $O(nk\log n)$ time. Updating the dual weights
and augmenting the matching each take $O(n)$ time.
Therefore, each request is processed in $O(nk\log n)$ time.
\end{proof}

\paragraph{Correctness.}
Initializing the new request with
dual weight zero preserves feasibility because server dual
weights are nonpositive.
Next, we show that the dual update and augmentation preserve
$(t,\alpha)$-feasibility and that the selected path minimizes
the $(t,\alpha)$-net-cost.  The following lemma establishes
feasibility after the remaining steps and shows that the
selected path has zero slack immediately before augmentation.

\begin{restatable}{lemma}{feasibilityLemma}\label{lem:feasibility}
The dual update and augmentation preserve
$(t,\alpha)$-feasibility. Immediately before augmentation,
every residual edge of $P$ has zero slack.
\end{restatable}

\begin{proof}
For the analysis, define
\[
\bar\ell(v)=\min\{\ell(v),\delta\},
\]
where $\bar\ell(v)=\delta$ if $v$ is unreachable.
The update adds $\delta-\bar\ell(v)$ to each request weight
and subtracts $\delta-\bar\ell(v)$ from each server weight.
Let $y'$ denote the updated weights, and let $\sigma$
denote the slacks before the update.

Consider a non-matching edge $(r',s')$, directed from
$r'$ to $s'$. The shortest-path inequality gives
\[
\ell(s')\le\ell(r')+\sigma(r',s').
\]
Since $\sigma(r',s')\ge0$, it follows that
\[
\bar\ell(s')-\bar\ell(r')\le\sigma(r',s').
\]
The change in its dual sum is
\[
\begin{aligned}
y'(r')+y'(s')-y(r')-y(s')
&=(\delta-\bar\ell(r'))
  -(\delta-\bar\ell(s'))\\
&=\bar\ell(s')-\bar\ell(r')\\
&\le\sigma(r',s').
\end{aligned}
\]
Thus, the dual sum increases by at most its slack, preserving
the upper-bound constraint.

Now consider a matching edge $(r',s')$, directed from
$s'$ to $r'$. The shortest-path inequality gives
\[
\ell(r')\le\ell(s')+\sigma(s',r'),
\]
and hence
\[
\bar\ell(r')-\bar\ell(s')\le\sigma(s',r').
\]
Therefore,
\[
\begin{aligned}
y'(r')+y'(s')-d(r',s')
&=\sigma(s',r')
  +\bar\ell(s')-\bar\ell(r')\\
&\ge0.
\end{aligned}
\]
This preserves the matching-edge lower bound.

We must also verify the upper bound on each matching edge.
The only incoming residual edge of the matched request
$r'$ is its edge from $s'$. Consequently, when these
vertices are reachable,
\[
\ell(r')=\ell(s')+\sigma(s',r')\ge\ell(s').
\]
Thus, $\bar\ell(r')\ge\bar\ell(s')$. If the vertices are
unreachable, both truncated distances equal $\delta$.
In either case,
\[
\begin{aligned}
y'(r')+y'(s')
&=y(r')+y(s')
  +\bar\ell(s')-\bar\ell(r')\\
&\le y(r')+y(s')\\
&\le t\,d(r',s').
\end{aligned}
\]
Hence the matching-edge upper bound is preserved.

For every free server $s'$, the zero-weight edge
$s'\to z$ gives
\[
\delta=\ell(z)\le\ell(s').
\]
Therefore, free servers receive no update and retain dual
weight zero. All server weights remain nonpositive because
they only decrease.

Along the selected shortest path, every residual edge
$u\to v$ satisfies
\[
\ell(v)=\ell(u)+\sigma(u,v).
\]
Every vertex on the path has distance at most $\delta$,
so $\bar\ell(v)=\ell(v)$ there. The slack after the update
is therefore
\[
\sigma(u,v)+\bar\ell(u)-\bar\ell(v)=0.
\]
Thus, every residual edge of $P$ has zero slack immediately
before augmentation.

It remains to verify feasibility after augmentation.
If $P$ is direct, its inserted edge has dual sum
$\alpha d(r,s)$. If $P$ is non-direct, it has no internal
free server, so none of its inserted edges joins two free
vertices before augmentation. Each such edge $(r',s')$
therefore has dual sum $t\,d(r',s')$. Since
$1\le\alpha\le t$, every inserted edge satisfies both
matching-edge bounds.

Every deleted edge $(r',s')$ has dual sum $d(r',s')$,
because its residual slack is zero. Its request remains
matched after augmentation, so its new upper bound is
$t\,d(r',s')$, which is satisfied.

Finally, augmentation matches the new request and may
saturate the terminal server. No request or server becomes
free. These changes can only relax an upper bound from
$\alpha d(r',s')$ to $t\,d(r',s')$. All remaining constraints
are unchanged. Therefore, feasibility is preserved.
\end{proof}

\begin{restatable}{corollary}{minPathCorollary}\label{cor:minpath}
The selected augmenting path $P$ has minimum
$(t,\alpha)$-net-cost.
\end{restatable}

\begin{proof}
Consider the dual weights and residual slacks immediately
after the update and before augmentation. Let $Q$ be any
augmenting path from the new request $r$ to a free server
$s'$. Summing the slack definitions along $Q$ cancels the
dual contributions of all internal vertices, giving
\[
\phi_{t,\alpha}(Q)
\ge
y(r)+y(s')+\sum_{e\in Q}\sigma(e).
\]
The inequality allows for a discounted first edge when a
non-direct path first visits a free server: the slack
definition uses coefficient $\alpha$ on that edge, whereas
the net-cost of a non-direct path uses coefficient $t$.
Since $y(s')=0$ and all residual slacks are nonnegative,
\[
\phi_{t,\alpha}(Q)\ge y(r).
\]

The selected path $P$ has no internal free server.
Consequently, the first inequality is an equality for $P$.
By Lemma~\ref{lem:feasibility}, all its residual edges have
zero slack. Hence
\[
\phi_{t,\alpha}(P)=y(r),
\]
which proves that $P$ has minimum $(t,\alpha)$-net-cost.
\end{proof}

\subsection{Properties used in the analysis}
We first show that direct edges retain their dual contribution while their server remains free.

\begin{restatable}{lemma}{persistenceLemma}\label{lem:persistence}
Suppose $\alpha>1$. Until a server $s$ becomes saturated, every positive-weight edge $(r,s)$ inserted
into $M^*$ by a direct augmenting path remains in $M^*$ and satisfies
\begin{equation}\label{eq:persistence}
y(r)+y(s)=\alpha d(r,s)
\end{equation}
\end{restatable}

\begin{proof}
On insertion, the direct edge has zero forward slack, so its dual sum is $\alpha d(r,s)$. Its backward slack is therefore
\[\gamma=(\alpha-1)d(r,s)>0.\]
Consider a subsequent search while $s$ is free. Its zero-length edge to $z$ implies $\ell(s)\ge\delta$. The only incoming residual edge of $r$ is its matching edge from $s$. If reachable, it therefore satisfies
\[\ell(r)=\ell(s)+\gamma>\delta.\]
Neither endpoint receives a dual update. Moreover, a shortest path of length $\delta$ cannot traverse this backward edge and delete the matching edge. Unreachable endpoints also receive no update and cannot occur on the selected path.

Applying this argument to successive arrivals proves the claim.
\end{proof}


We next bound the weight of the offline matching by the dual objective.
\begin{restatable}{lemma}{dualBoundLemma}\label{lem:dualbound}
After processing any arrival,
\begin{equation}\label{eq:dualbound}
w(M^*)\le\sum_r y(r)+\sum_s b(s)y(s)\le t\,\OPT,
\end{equation}
where the first sum ranges over the requests that have arrived.
\end{restatable}

\begin{proof}
Let $\deg_{M^*}(s)$ be the number of requests assigned to server $s$. Summing the matching-edge lower bounds gives
\[w(M^*)\le\sum_r y(r)+\sum_s\deg_{M^*}(s)y(s).\]
If $s$ is saturated, then $\deg_{M^*}(s)=b(s)$. Otherwise, $y(s)=0$. Thus,
\[w(M^*)\le\sum_r y(r)+\sum_s b(s)y(s).\]
Fix a final optimal matching, and let $B$ be its restriction to the requests that have arrived. Since $\deg_B(s)\le b(s)$ and $y(s)\le0$,
\[
\begin{aligned}
\sum_r y(r)+\sum_s b(s)y(s)
&\le\sum_r y(r)+\sum_s\deg_B(s)y(s)\\
&=\sum_{(r,s)\in B}\bigl(y(r)+y(s)\bigr)\\
&\le t\,w(B)\le t\,\OPT.
\end{aligned}
\]
The penultimate inequality follows from the pairwise upper bounds.
\end{proof}

\section{Competitive ratio}
\label{sec:competitive-ratio}

For the analysis, index requests by their arrival order. Let $M_i$ and $M_i^*$ be the online and offline matchings after arrival $i$, let $P_i$ be the selected path, and let $d_i$ be the distance from $r_i$ to a nearest free server before it is served. Write $R_i=\{r_1,\ldots,r_i\}$ and $\ALG=w(M_n)$. In this section we establish that the competitive ratio of RG is upper bounded by $20k/3 - 3$ for $t=4, \alpha = 2$. 



\textbf{Saturation Epochs.} A saturation event occurs when a server uses its last unit of available capacity. Each arrival increases the load of exactly one server by one. Hence each server becomes saturated once, and at most one server becomes saturated on an arrival. Let
$
0=\tau_0<\tau_1<\cdots<\tau_k=n
$
be the saturation times. Epoch $h$ consists of arrivals $\tau_{h-1}+1,\ldots,\tau_h$. Fix an epoch and write $a=\tau_{h-1}$ and $b=\tau_h$. Immediately before every arrival in this epoch, every server that was free at time $a$ is still free.

Let
\[
G=\sum_{\substack{a<i\le b\\P_i\text{ direct}}}d_i,
\qquad
\Phi=\sum_{i=a+1}^{b}\phi_{t,\alpha}^{M_{i-1}^*}(P_i).
\]


\begin{restatable}[Direct-cost bound]{lemma}{directCostLemma}\label{lem:greedybudget}
For every epoch,
$w(M_b^*)+(\alpha-1)G\le t\,\OPT$
\end{restatable}

\begin{proof}
For $\alpha=1$, the claim follows from
Lemma~\ref{lem:dualbound}. Assume $\alpha>1$.

For every positive-length direct edge $e=(r,s)$ inserted during the
epoch, Lemma~\ref{lem:persistence} implies that $e\in M_b^*$ at the
end of the epoch and $
y(r)+y(s)=\alpha d(e).
$
These edges have total weight $G$. Every remaining edge
$e'=(r',s')\in M_b^*$ satisfies
$
y(r')+y(s')\ge d(e')
$
by the matching-edge lower bound~\eqref{eq:lower}. Therefore,
\[
\begin{aligned}
w(M_b^*)+(\alpha-1)G
&\le \sum_{(r,s)\in M_b^*}\!\bigl(y(r)+y(s)\bigr)\\
&=\sum_{r\in R_b}y(r)+\sum_s\deg_{M_b^*}(s)y(s)
 =\sum_{r\in R_b}y(r)+\sum_s b(s)y(s).
\end{aligned}
\]
The last equality holds because an unsaturated server has $y(s)=0$,
whereas a saturated server has
$\deg_{M_b^*}(s)=b(s)$. Lemma~\ref{lem:dualbound} bounds the final
expression by $t\,\OPT$.
\end{proof}

\begin{restatable}[Nearest-free-distance bound]{lemma}{nearestFreeDistanceLemma}\label{lem:distance}
For every epoch,
$\sum_{i=a+1}^{b}d_i\le\OPT+w(M_a^*)$
\end{restatable}
\begin{proof}
Fix a final optimal matching and let $B$ be its restriction to $R_b$. In $M_a^*\oplus B$, direct the edges of $B\setminus M_a^*$ from requests to servers and the edges of $M_a^*\setminus B$ from servers to requests. Each old request has equal indegree and outdegree, whereas each new request has one outgoing edge and no incoming edge. At a server $s$, indegree minus outdegree equals $\deg_B(s)-\deg_{M_a^*}(s)$.

Decompose this directed graph into edge-disjoint paths from vertices with excess outgoing edges to vertices with excess incoming edges, together with cycles. This can be obtained by pairing incoming and outgoing edges at each vertex and removing cycles. Each new request $r_i$ starts one such path, which must end at a server $s_i$ with $\deg_B(s_i)>\deg_{M_a^*}(s_i)$. Since $\deg_B(s_i)\le b(s_i)$, that server is free at time $a$.

No server that is free at time $a$ becomes saturated before the last arrival of the epoch. Thus $s_i$ is available immediately before arrival $i$. The triangle inequality gives
\[d_i\le d(r_i,s_i)\le w(\text{path from }r_i\text{ to }s_i).\]
Summing over the new requests and using edge-disjointness gives
$\sum_{i=a+1}^b d_i\le w(M_a^*)+w(B)\le w(M_a^*)+\OPT$.
\end{proof}

A direct edge to a nearest free server has net-cost $\alpha d_i$, so minimum-net-cost selection gives $\phi_{t,\alpha}(P_i)\le\alpha d_i$. A selected direct path consequently has weight $d_i$. Summing and applying the nearest-free-distance bound gives
\begin{equation}\label{eq:modifiedbudget}
\Phi\le\alpha(\OPT+w(M_a^*)).
\end{equation}


Set $t=4$ and $\alpha=2$. Lemma~\ref{lem:greedybudget} and~\eqref{eq:modifiedbudget} give
\[
w(M_b^*)+G\le4\,\OPT,
\qquad
\Phi\le2(\OPT+w(M_a^*)).
\]
Consequently,
\begin{equation}\label{eq:combined}
\begin{aligned}
\Phi+2G
&\le2(\OPT+w(M_a^*))+2(4\,\OPT-w(M_b^*))\\
&=10\,\OPT+2(w(M_a^*)-w(M_b^*)).
\end{aligned}
\end{equation}
Applying Lemma~\ref{lem:accounting} with $t=4$ and $\alpha=2$, the online cost incurred during the epoch satisfies
\[
w(M_b)-w(M_a)
\le\frac{2\Phi+4G+5\bigl(w(M_a^*)-w(M_b^*)\bigr)}{3}.
\]
Using~\eqref{eq:combined}, we obtain
\[
w(M_b)-w(M_a)
\le\frac{20}{3}\OPT+3\bigl(w(M_a^*)-w(M_b^*)\bigr).
\]
Summing this inequality over the $k$ epochs and using $w(M_0^*)=0$ and $w(M_n^*)\ge\OPT$ gives
\[
\ALG
\le \frac{20}{3}k\,\OPT
+3\bigl(w(M_0^*)-w(M_n^*)\bigr)
\le \left(\frac{20}{3}k-3\right)\OPT.
\]

\subsection{Tighter Analysis Using Parameter Optimization}
\label{appdx:parametric}

The epoch bounds hold for every $t>1$ and $1\le\alpha\le t$. Choosing $\alpha=\sqrt t$ makes the offline matching terms telescope with equal coefficients.

\begin{corollary}\label{cor:optimized}
For $t>1$ and $\alpha=\sqrt t$,
\begin{equation}\label{eq:optimized}
\ALG\le
\left[
\frac{2\sqrt t(t+1)}{t-1}k
-\frac{t+1+2\sqrt t}{t-1}
\right]\OPT.
\end{equation}
\end{corollary}
\begin{proof}
Since $\alpha>1$, the direct-cost bound implies
\[
G\le\frac{t\,\OPT-w(M_b^*)}{\alpha-1}.
\]
For $t=\alpha^2$,
\[
\frac{t-\alpha}{\alpha-1}=\alpha.
\]
Combining this with~\eqref{eq:modifiedbudget},
\[
\begin{aligned}
\Phi+(t-\alpha)G
&\le\alpha(\OPT+w(M_a^*))+\alpha(t\,\OPT-w(M_b^*))\\
&=\alpha(t+1)\OPT+\alpha(w(M_a^*)-w(M_b^*)).
\end{aligned}
\]
Substitution into Lemma~\ref{lem:accounting} gives
\[
(t-1)\bigl(w(M_b)-w(M_a)\bigr)
\le2\alpha(t+1)\OPT+(2\alpha+t+1)(w(M_a^*)-w(M_b^*)).
\]
Sum over the epochs and use $w(M_0^*)=0$ and $w(M_n^*)\ge\OPT$.
\end{proof}

The coefficient of $k$ in~\eqref{eq:optimized} is
\[
K(t)=\frac{2\sqrt t(t+1)}{t-1}.
\]
Its derivative is
\[
K'(t)=\frac{t^2-4t-1}{\sqrt t(t-1)^2}.
\]
Hence the minimum over $t>1$ occurs at
\[
t=2+\sqrt5,\qquad \alpha=\sqrt{2+\sqrt5}.
\]
For these parameters, the competitive ratio in~\eqref{eq:optimized} is approximately
\[
6.660381353571k-2.890053638264.
\]

This proves Theorem~\ref{thm:main}.

\paragraph{Proof of Theorem~\ref{thm:combined}.}
We first prove the competitive guarantee for the full parameter range in the theorem. Fix an epoch, write $W_{h-1}=w(M_a^*)$ and $W_h=w(M_b^*)$, and suppose first that $\alpha>1$. The direct-cost and nearest-free-distance bounds give
\[
G\le \frac{t\,\OPT-W_h}{\alpha-1},
\qquad
\Phi\le \alpha(\OPT+W_{h-1}).
\]
Substituting these inequalities into Lemma~\ref{lem:accounting}, the epoch cost $C_h=w(M_b)-w(M_a)$ satisfies
\[
\begin{aligned}
(t-1)C_h
&\le
\left(2\alpha+\frac{2t(t-\alpha)}{\alpha-1}\right)\OPT
+ A W_{h-1}-B W_h,
\end{aligned}
\]
where
\[
A=2\alpha+t+1,
\qquad
B=\frac{2(t-\alpha)}{\alpha-1}+t+1.
\]
For $\alpha\le\sqrt t$,
\[
A-B=\frac{2(\alpha^2-t)}{\alpha-1}\le0.
\]
Summing over the $k$ epochs, using $W_0=0$, and dropping the resulting nonpositive intermediate and terminal terms yields
\[
\ALG\le
\frac{2\alpha+2t(t-\alpha)/(\alpha-1)}{t-1}\,k\OPT
=O(k)\OPT.
\]

If $\alpha=1$, then $G\le\sum_{i=a+1}^b d_i\le\OPT+W_{h-1}$ and $\Phi\le\OPT+W_{h-1}$. Lemma~\ref{lem:accounting} and Lemma~\ref{lem:dualbound} therefore give
\[
(t-1)C_h
\le 2t\OPT+(3t+1)W_{h-1}-(t+1)W_h
\le 3t(t+1)\OPT.
\]
Summing again proves the $O(k)$ guarantee.

It remains to prove the greedy-cost fraction. If $t>3+2\sqrt2$, then $\sqrt t<(t-1)/2$, so every $1\le\alpha\le\sqrt t$ satisfies the hypothesis of Lemma~\ref{lem:greedyfraction}. At $t=3+2\sqrt2$, the same is true for $\alpha<\sqrt t$; when $\alpha=\sqrt t=1+\sqrt2$, the claimed fraction equals zero and the conclusion follows from $G\ge0$. This completes the proof.

\section{Competitive Analysis for the RM Algorithm}
\label{appdx:rm-analysis}
In this section, we prove the competitive ratio of the RM algorithm.

 \begin{theorem}[Competitive ratio of RM]
\label{thm:rm-competitive}
For $t=1+\sqrt{2}$, the RM algorithm is
\[
(6+4\sqrt{2})k<11.66k
\]
competitive for online metric transportation.
\end{theorem}

\begin{proof}
RM is the special case of RG with $\alpha=t$. Retain the saturation
epochs from Section~\ref{sec:competitive-ratio}, and write
\[
W_h=w(M_{\tau_h}^*),\qquad W_0=0.
\]
For epoch $h$, let
\[
\Phi_h
=
\sum_{i=\tau_{h-1}+1}^{\tau_h}
\phi_{t,t}^{M_{i-1}^*}(P_i).
\]
A direct edge from $r_i$ to a nearest free server is a candidate path
of net-cost $t d_i$. Hence, by minimum-net-cost selection and
Lemma~\ref{lem:distance},
\[
\Phi_h
\le t\sum_{i=\tau_{h-1}+1}^{\tau_h}d_i
\le t(\OPT+W_{h-1}).
\]

Since $\alpha=t$, the greedy-cost term in
Lemma~\ref{lem:accounting} vanishes. Thus, the online cost $C_h$
incurred during epoch $h$ satisfies
\[
\begin{aligned}
C_h
&\le
\frac{2\Phi_h+(t+1)(W_{h-1}-W_h)}{t-1}\\
&\le
\frac{2t\,\OPT+(3t+1)W_{h-1}-(t+1)W_h}{t-1}.
\end{aligned}
\]
Summing over the $k$ epochs gives
\[
\ALG
\le
\frac{
2tk\,\OPT
+2t\sum_{h=1}^{k-1}W_h
-(t+1)W_k
}{t-1}.
\]
By Lemma~\ref{lem:dualbound}, $W_h\le t\,\OPT$ for every $h$, while
$W_k\ge\OPT$ because $M_{\tau_k}^*$ is a feasible matching of all
requests. Therefore,
\[
\ALG
\le
\left[
\frac{2t(t+1)}{t-1}k
-\frac{2t^2+t+1}{t-1}
\right]\OPT.
\]
The leading coefficient
\[
\frac{2t(t+1)}{t-1}
\]
is minimized over $t>1$ at $t=1+\sqrt{2}$. Substituting this value
yields
\[
\ALG
\le
\bigl((6+4\sqrt{2})k-(5+4\sqrt{2})\bigr)\OPT
\le
(6+4\sqrt{2})k\,\OPT.
\]
\end{proof}

\section{Input-sensitive guarantee}
\label{appdx:metric-sensitive}

The preceding competitive analysis depends only on the number of server
locations. We now give a complementary guarantee that adapts to their
geometry. For $S'\subseteq S$, let
$\diam(S')=\max_{u,v\in S'}d(u,v)$, and let $\tsp(S')$ be the length of a
minimum closed tour through $S'$. Define
\begin{equation}
  \mu_X(S)
  =\max_{\substack{S'\subseteq S\\ \diam(S')>0}}
       \frac{\tsp(S')}{\diam(S')}.
  \label{eq:is-mu}
\end{equation}
When $\diam(S)=0$, we set $\mu_X(S)=1$.

\begin{theorem}[Restatement of Theorem~\ref{thm:metric-sensitive}]
\label{thm:input-sensitive-restated}
\[
  \ALG
  \leq O_{t,\alpha}\!\left(\mu_X(S)\log^2 n\right)\OPT.
\]
\end{theorem}

For server locations in Euclidean spaces of dimension $d>1$,
$\mu_X(S)=O_d(k^{1-1/d})$, yielding the desired competitive ratio in Theorem~\ref{thm:metric-sensitive}.

We use the notation introduced above. Thus, $M_i$ and $M_i^*$ are the online
and offline matchings after request $r_i$ is processed, $P_i$ is the selected
augmenting path, and
\[
  \phi_i=\phi_{t,\alpha}^{M_{i-1}^*}(P_i).
\]
In particular, $\ALG=w(M_n)$. Fix a minimum-cost feasible matching
$M_{\mathrm{OPT}}$ of all requests, so that
$w(M_{\mathrm{OPT}})=\OPT$, and let $\opt(r)$ be the server to which $r$ is
assigned in $M_{\mathrm{OPT}}$. For a request set $R'\subseteq R$, let
$\sigma(R')$ denote its requests in arrival order.

\subsection{Algorithmic facts used by the geometric argument}

We first isolate the three consequences of the current primal--dual
implementation that replace the corresponding RM properties in
\citet{Nayyar2017Mu}.

\begin{lemma}[Basic net-cost bounds]
\label{lem:is-basic}
For every arrival $i$,
\[
  0\leq \phi_i\leq t\OPT,
  \qquad
  w(M_i^*)\leq t\OPT.
\]
Moreover, immediately before augmentation in phase $i$,
\[
  \phi_i=y(r_i).
\]
\end{lemma}

\begin{proof}
The equality $\phi_i=y(r_i)$ follows from Lemma~\ref{lem:feasibility} and
Corollary~\ref{cor:minpath}: after the dual update, every residual edge of the selected path has zero
slack. Request weights are initialized to zero and only increase, so
$\phi_i=y(r_i)\geq0$. The bound
\[
  w(M_i^*)\leq t\OPT
\]
is Lemma~\ref{lem:dualbound}.

It remains to prove $\phi_i\leq t\OPT$. Orient the edges of
\[
  M_{\mathrm{OPT}}\mathbin{\triangle}M_{i-1}^*
\]
from requests to servers when they belong to $M_{\mathrm{OPT}}$, and from
servers to requests when they belong to $M_{i-1}^*$. Pair incoming and
outgoing edges at requests and servers, and consider the resulting directed
trail beginning at the new request $r_i$. Repeatedly erase any closed directed
subtrail. The result is a simple alternating path $Q$ ending at a server whose
$M_{i-1}^*$-load is smaller than its $M_{\mathrm{OPT}}$-load, and hence at a
free server. Truncate $Q$ at the first free server it encounters.

If $Q$ is direct, then
\[
  \phi_{t,\alpha}^{M_{i-1}^*}(Q)
  =\alpha w(Q)
  \leq t\,w(M_{\mathrm{OPT}}\cap Q)
  \leq t\OPT.
\]
If $Q$ is non-direct, then
\[
  \phi_{t,\alpha}^{M_{i-1}^*}(Q)
  =t\,w(Q\setminus M_{i-1}^*)-w(Q\cap M_{i-1}^*)
  \leq t\,w(M_{\mathrm{OPT}}\cap Q)
  \leq t\OPT.
\]
By Corollary~\ref{cor:minpath}, $P_i$ is a minimum-$(t,\alpha)$-net-cost
augmenting path. Therefore
\[
  \phi_i
  \leq\phi_{t,\alpha}^{M_{i-1}^*}(Q)
  \leq t\OPT.
\]
\end{proof}

\begin{lemma}[Online cost from net cost]
\label{lem:is-online-from-net}
Let
\[
  \Phi=\sum_{i=1}^n\phi_i.
\]
Then
\[
  w(M_n)\leq \frac{2t}{\alpha(t-1)}\,\Phi.
\]
\end{lemma}

\begin{proof}
Let $G$ be the total weight of the selected direct paths. The total
contribution of these paths to $\Phi$ is $\alpha G$. Every selected
non-direct path has nonnegative net cost by
Lemma~\ref{lem:is-basic}. Consequently,
\[
  \alpha G\leq\Phi.
\]

Apply Lemma~\ref{lem:accounting} to the complete execution. Since the initial
offline matching is empty, it gives
\[
  w(M_n)
  \leq
  \frac{2\Phi+2(t-\alpha)G-(t+1)w(M_n^*)}{t-1}.
\]
Dropping the nonpositive terminal term and using
$G\leq\Phi/\alpha$, we obtain
\[
  w(M_n)
  \leq
  \frac{2\Phi+2(t-\alpha)\Phi/\alpha}{t-1}
  =
  \frac{2t}{\alpha(t-1)}\,\Phi.
\]
\end{proof}

\begin{lemma}[Free-server separation]
\label{lem:is-free-separation}
Suppose request $r$ is about to be processed and $s$ is free with respect to
the current offline matching. Then
\[
  d(r,s)\geq\frac{\phi_{t,\alpha}(r)}{\alpha},
\]
where $\phi_{t,\alpha}(r)$ denotes the net cost selected while processing $r$.
\end{lemma}

\begin{proof}
Immediately before augmentation, $r$ and $s$ are both free. By
Lemma~\ref{lem:is-basic}, $\phi_{t,\alpha}(r)=y(r)$. The free--free dual
constraint and the free-server normalization therefore give
\[
  \phi_{t,\alpha}(r)
  =y(r)
  \leq \alpha d(r,s)-y(s)
  =\alpha d(r,s).
\]
\end{proof}

The denominator in Lemma~\ref{lem:is-free-separation} is $\alpha$, rather than
$t$ as in the RM proof, because Robustified Greedy discounts direct
free--free edges.

\subsection{Dyadic net-cost groups}

Assume henceforth that $\OPT>0$. If $\OPT=0$, Lemmas
\ref{lem:is-basic} and~\ref{lem:is-online-from-net} immediately give
$w(M_n)=0$.

Partition the requests into $R'$ and $R''$ as follows. Put $r_i\in R'$ if at
least one of
\begin{equation}
  \phi_i\leq\frac{\OPT}{n},
  \qquad
  \phi_i\leq16\alpha\,d(r_i,\opt(r_i))
  \label{eq:is-easy}
\end{equation}
holds, and put all remaining requests in $R''$. Then
\begin{equation}
  \sum_{r_i\in R'}\phi_i
  \leq(16\alpha+1)\OPT.
  \label{eq:is-easy-sum}
\end{equation}
Indeed, charge the requests satisfying the first condition to the $n$ copies
of $\OPT/n$. Charge every remaining request in $R'$ to its distinct edge in
$M_{\mathrm{OPT}}$ using the second condition.

Partition $R''$ into groups $R_1,\ldots,R_m$, where $r_i\in R_j$ if
\begin{equation}
  \frac{2^{j-1}\OPT}{n}
  \leq\phi_i
  <\frac{2^j\OPT}{n}.
  \label{eq:is-net-group}
\end{equation}
By Lemma~\ref{lem:is-basic}, we may take
\begin{equation}
  m=1+\left\lceil\log_2(tn)\right\rceil.
  \label{eq:is-number-net-groups}
\end{equation}
We will prove that, for every $j$,
\begin{equation}
  \sum_{r_i\in R_j}\phi_i
  \leq
  O_{t,\alpha}\!\left(\mu_X(S)\log(2n)\right)\OPT.
  \label{eq:is-one-net-group-goal}
\end{equation}

\subsection{Clusters and optimal paths}

Fix a group $R_j$, write
\[
  \sigma(R_j)=(r_1',\ldots,r_N'),
\]
and define
\begin{equation}
  r_{\mathrm{IB}}^j
  =\frac{2^{j-1}\OPT}{\alpha n}.
  \label{eq:is-inner-radius}
\end{equation}
Process $\sigma(R_j)$ in arrival order. The first request starts a cluster and
is its center. Subsequently, assign a request to an existing cluster if its
distance from the nearest center is less than $r_{\mathrm{IB}}^j/2$;
otherwise, start a new cluster with that request as its center. Denote the
resulting clusters by $C_1,\ldots,C_{\widetilde{k}}$, and denote the center of
$C_k$ by $c^{(k)}$. The construction gives
\begin{align}
  d(r,c^{(k)})
  &<\frac{r_{\mathrm{IB}}^j}{2}
  &&\text{for every }r\in C_k,
  \label{eq:is-cluster-radius}\\
  d(c^{(k)},c^{(k')})
  &\geq\frac{r_{\mathrm{IB}}^j}{2}
  &&\text{for }k\neq k'.
  \label{eq:is-center-separation}
\end{align}

Let $M^{(k)}$ be the offline matching immediately before $c^{(k)}$ arrives.
Because $c^{(k)}$ is the first request of $C_k$ in arrival order, every request
of $C_k$ is free with respect to $M^{(k)}$.

Orient
\[
  M_{\mathrm{OPT}}\mathbin{\triangle}M^{(k)}
\]
as in the proof of Lemma~\ref{lem:is-basic}, and fix one edge-disjoint
decomposition into directed alternating trails and cycles. Repeatedly erase
closed directed subtrails from each trail. Every arrived request is balanced,
every unarrived request has one unit of excess outdegree, and a server $s$ has
indegree excess
\[
  b(s)-\deg_{M^{(k)}}(s)\geq0,
\]
because $M_{\mathrm{OPT}}$ saturates every server. Thus every
$r_g\in C_k$ is a source of the resulting decomposition, and its path ends at
a server with positive indegree excess, equivalently at a server $s$ with
\[
  \deg_{M^{(k)}}(s)<b(s).
\]
The terminal server is therefore free with respect to $M^{(k)}$.

Let $Q_g$ be this path, truncated when it first reaches a free server. We call
$Q_g$ the \emph{optimal path} of $r_g$. Define
\begin{equation}
  \delta_g
  =\max_{v\in V(Q_g)}d(c^{(k)},v).
  \label{eq:is-delta}
\end{equation}
The paths $Q_g$, for $r_g\in C_k$, are edge-disjoint. This is the capacitated
analogue of the vertex-disjoint path family used by \citet{Nayyar2017Mu}.

\begin{lemma}[Radius of an optimal path]
\label{lem:is-delta-range}
For every $r_g\in C_k$,
\[
  r_{\mathrm{IB}}^j
  \leq\delta_g
  \leq
  \left(t+1+\frac{t}{2\alpha}\right)\OPT.
\]
\end{lemma}

\begin{proof}
The terminal server of $Q_g$ is free with respect to $M^{(k)}$. By
Lemma~\ref{lem:is-free-separation}, every such server $s$ satisfies
\[
  d(c^{(k)},s)
  \geq\frac{\phi_{t,\alpha}(c^{(k)})}{\alpha}
  \geq\frac{2^{j-1}\OPT}{\alpha n}
  =r_{\mathrm{IB}}^j.
\]
This proves the lower bound.

The path $Q_g$ uses only edges of $M_{\mathrm{OPT}}$ and $M^{(k)}$. Therefore,
by Lemma~\ref{lem:is-basic},
\[
  w(Q_g)
  \leq w(M_{\mathrm{OPT}})+w(M^{(k)})
  \leq(t+1)\OPT.
\]
Furthermore, by~\eqref{eq:is-cluster-radius},
\[
  d(c^{(k)},r_g)
  <\frac{r_{\mathrm{IB}}^j}{2}
  \leq\frac{\phi_{t,\alpha}(r_g)}{2\alpha}
  \leq\frac{t}{2\alpha}\OPT.
\]
The triangle inequality now gives the upper bound.
\end{proof}

\begin{lemma}[Optimal-path density]
\label{lem:is-path-density}
For every $r_g\in C_k$,
\begin{align}
  w(Q_g)
  &\geq\frac{\delta_g}{2},
  \label{eq:is-path-length}\\
  w(M_{\mathrm{OPT}}\cap Q_g)
  &\geq\frac{w(Q_g)}{t+1}.
  \label{eq:is-opt-fraction}
\end{align}
\end{lemma}

\begin{proof}
By the triangle inequality,
\[
  \delta_g
  \leq d(c^{(k)},r_g)+w(Q_g).
\]
Equations~\eqref{eq:is-cluster-radius} and the lower bound in
Lemma~\ref{lem:is-delta-range} imply
\[
  d(c^{(k)},r_g)<\frac{\delta_g}{2},
\]
which proves~\eqref{eq:is-path-length}.

For~\eqref{eq:is-opt-fraction}, the claim is immediate when $Q_g$ is direct,
because its only edge belongs to $M_{\mathrm{OPT}}$. Suppose therefore that
$Q_g$ is non-direct. Consider the feasible dual vector associated with
$M^{(k)}$, extended by assigning dual weight zero to requests that have not
yet arrived. This extension remains feasible because all server weights are
nonpositive.

Since $\alpha\leq t$, the free--free constraint is stronger than the ordinary
$t$-constraint. Consequently, every forward edge $(r,s)$ of $Q_g$ satisfies
\[
  y(r)+y(s)\leq t\,d(r,s),
\]
regardless of whether its endpoints are free. Every backward matching edge
satisfies
\[
  y(r)+y(s)\geq d(r,s).
\]
Summing these inequalities along $Q_g$ cancels the internal dual weights. Its
initial request has dual weight zero, while its terminal server is free and
therefore also has dual weight zero. Hence
\[
  t\,w(Q_g\setminus M^{(k)})
  -w(Q_g\cap M^{(k)})
  \geq0.
\]

The forward edges belong to $M_{\mathrm{OPT}}$, and the backward edges belong
to $M^{(k)}$. Write
\[
  A=w(M_{\mathrm{OPT}}\cap Q_g),
  \qquad
  B=w(M^{(k)}\cap Q_g).
\]
The preceding inequality is $tA\geq B$. Therefore
\[
  w(Q_g)=A+B\leq(t+1)A,
\]
which proves~\eqref{eq:is-opt-fraction}.
\end{proof}

A direct optimal path has optimal-edge fraction one. For a non-direct optimal
path, the ordinary $t$-bound remains valid on every forward edge because
$\alpha\leq t$.

\subsection{Outer groups and weighted balls}

Partition the requests of $R_j$ into outer groups $R^l$, for $l\geq j$, by
placing $r_g$ in $R^l$ when
\begin{equation}
  \frac{2^{l-1}\OPT}{\alpha n}
  \leq\delta_g
  <\frac{2^l\OPT}{\alpha n}.
  \label{eq:is-outer-group}
\end{equation}
Lemma~\ref{lem:is-delta-range} implies that there are at most
\begin{equation}
  L
  =
  1+\left\lceil
       \log_2\!\left(
         \alpha n\left(t+1+\frac{t}{2\alpha}\right)
       \right)
     \right\rceil
  \label{eq:is-number-outer-groups}
\end{equation}
nonempty outer-group indices.

Fix $l$. Let
\[
  C_k^l=C_k\cap R^l,
  \qquad
  \beta_k^l=|C_k^l|,
\]
omitting empty sets. Associate with $C_k^l$ two concentric weighted balls
centered at $c^{(k)}$:
\begin{align*}
  \mathrm{IB}_k^l
  &=B(c^{(k)},r_{\mathrm{IB}}^j),
  &\text{weight }&\beta_k^l,\\
  \mathrm{OB}_k^l
  &=B(c^{(k)},r_{\mathrm{OB}}^l),
  &r_{\mathrm{OB}}^l&=\frac{2^l\OPT}{\alpha n},
  \qquad \text{weight }\beta_k^l.
\end{align*}
Because $l\geq j$,
\[
  r_{\mathrm{IB}}^j<r_{\mathrm{OB}}^l.
\]
Moreover, every optimal path $Q_g$ with $r_g\in C_k^l$ is contained in
$\mathrm{OB}_k^l$.

\begin{lemma}[Inner-ball bounds]
\label{lem:is-inner-balls}
For the fixed outer group $R^l$,
\begin{align}
  \sum_{r_i\in R^l}\phi_i
  &<2\alpha\sum_k\beta_k^l r_{\mathrm{IB}}^j,
  \label{eq:is-inner-pays-net}\\
  |\widetilde S|\frac{r_{\mathrm{IB}}^j}{4}
  &\leq\tsp(\widetilde S)
  \label{eq:is-inner-pays-tsp}
\end{align}
for every set $\widetilde R$ of at least two cluster centers, where
\[
  \widetilde S=\opt(\widetilde R).
\]
\end{lemma}

\begin{proof}
For every $r_i\in R_j$, equations~\eqref{eq:is-net-group} and
\eqref{eq:is-inner-radius} give
\[
  \frac{\phi_i}{2\alpha}<r_{\mathrm{IB}}^j.
\]
Summing over the partition of $R^l$ into the sets $C_k^l$ proves
\eqref{eq:is-inner-pays-net}.

Now consider two distinct centers $c^{(k)}$ and $c^{(k')}$. Since both belong
to $R''$, equations~\eqref{eq:is-easy} and~\eqref{eq:is-net-group} imply
\[
  d(c^{(k)},\opt(c^{(k)}))
  <\frac{\phi_{t,\alpha}(c^{(k)})}{16\alpha}
  <\frac{r_{\mathrm{IB}}^j}{8},
\]
and the same inequality holds for $c^{(k')}$. Combining this with
\eqref{eq:is-center-separation} gives
\[
  d(\opt(c^{(k)}),\opt(c^{(k')}))
  >\frac{r_{\mathrm{IB}}^j}{4}.
\]
In particular, different centers in $\widetilde R$ have different optimal
server locations. Every edge of a tour through $\widetilde S$ therefore has
length at least $r_{\mathrm{IB}}^j/4$, proving
\eqref{eq:is-inner-pays-tsp}.
\end{proof}

For a ball $B\subseteq X$, write $M_{\mathrm{OPT}}\cap B$ for the set of
optimal edges both of whose endpoints lie in $B$.

\begin{lemma}[Optimal weight in an outer ball]
\label{lem:is-outer-density}
For every nonempty $C_k^l$,
\[
  w(M_{\mathrm{OPT}}\cap\mathrm{OB}_k^l)
  \geq
  \frac{\beta_k^l r_{\mathrm{OB}}^l}{4(t+1)}.
\]
\end{lemma}

\begin{proof}
For every $r_g\in C_k^l$, equation~\eqref{eq:is-outer-group} and
Lemma~\ref{lem:is-path-density} give
\[
  w(M_{\mathrm{OPT}}\cap Q_g)
  \geq\frac{w(Q_g)}{t+1}
  \geq\frac{\delta_g}{2(t+1)}
  \geq\frac{r_{\mathrm{OB}}^l}{4(t+1)}.
\]
The paths associated with the requests of $C_k^l$ are edge-disjoint and lie
inside $\mathrm{OB}_k^l$. Summing the displayed inequality over these
$\beta_k^l$ paths proves the claim.
\end{proof}

\subsection{Weighted Vitali selection}

We use the same weighted variant of the Vitali selection as
\citet{Nayyar2017Mu}. Given a collection of equal-radius weighted balls,
repeatedly select a remaining ball of maximum weight and remove it together
with every remaining ball that intersects it. The selected balls are pairwise
disjoint, and their associated intersecting sets partition the original
collection. Every ball in an intersecting set lies inside the threefold
expansion of its selected ball, and its weight is no larger than the weight of
the selected ball.

Apply this procedure to the outer balls $\{\mathrm{OB}_k^l\}_k$. Write the
selected balls as
\[
  \mathrm{OB}_{s_1}^l,\ldots,\mathrm{OB}_{s_q}^l,
\]
and let $K_{s_p}$ be the indices of the balls assigned to
$\mathrm{OB}_{s_p}^l$. Set
\[
  \kappa_{s_p}
  =\frac{\beta_{s_p}^l r_{\mathrm{OB}}^l}{4(t+1)},
  \qquad
  \theta_{s_p}
  =|K_{s_p}|\beta_{s_p}^l r_{\mathrm{IB}}^j.
\]
Lemma~\ref{lem:is-outer-density} and the disjointness of the selected outer
balls give
\begin{equation}
  \sum_{p=1}^q\kappa_{s_p}\leq\OPT.
  \label{eq:is-kappa}
\end{equation}
The maximal-weight selection and the partition into intersecting sets give
\begin{equation}
  \sum_k\beta_k^l r_{\mathrm{IB}}^j
  \leq\sum_{p=1}^q\theta_{s_p}.
  \label{eq:is-theta}
\end{equation}

\begin{lemma}[One outer group]
\label{lem:is-one-outer-group}
For every $j$ and $l$,
\[
  \sum_{r_i\in R^l}\phi_i
  \leq192\alpha(t+1)\mu_X(S)\OPT.
\]
\end{lemma}

\begin{proof}
We first bound $\theta_{s_p}/\kappa_{s_p}$. If $|K_{s_p}|=1$, then
$l\geq j$ gives
\[
  \frac{\theta_{s_p}}{\kappa_{s_p}}
  =
  4(t+1)\frac{r_{\mathrm{IB}}^j}{r_{\mathrm{OB}}^l}
  \leq2(t+1).
\]

Suppose $|K_{s_p}|>1$. Let $\widetilde R_{s_p}$ be the centers of the balls
indexed by $K_{s_p}$, and let
\[
  \widetilde S_{s_p}=\opt(\widetilde R_{s_p}).
\]
By Lemma~\ref{lem:is-inner-balls},
\[
  |K_{s_p}|r_{\mathrm{IB}}^j
  \leq4\tsp(\widetilde S_{s_p}).
\]
Furthermore,
\[
  \opt(c^{(k)})
  \in\mathrm{IB}_k^l
  \subseteq\mathrm{OB}_k^l
\]
for every $k$: the proof of Lemma~\ref{lem:is-inner-balls} gives
\[
  d(c^{(k)},\opt(c^{(k)}))
  <\frac{r_{\mathrm{IB}}^j}{8}.
\]
Every ball indexed by $K_{s_p}$ lies inside
$3\mathrm{OB}_{s_p}^l$. Hence
\[
  \diam(\widetilde S_{s_p})
  \leq6r_{\mathrm{OB}}^l.
\]
It follows that
\begin{align*}
  \frac{\theta_{s_p}}{\kappa_{s_p}}
  &=
  4(t+1)
  \frac{|K_{s_p}|r_{\mathrm{IB}}^j}{r_{\mathrm{OB}}^l}\\
  &\leq
  16(t+1)
  \frac{\tsp(\widetilde S_{s_p})}{r_{\mathrm{OB}}^l}\\
  &\leq
  96(t+1)
  \frac{\tsp(\widetilde S_{s_p})}
       {\diam(\widetilde S_{s_p})}\\
  &\leq
  96(t+1)\mu_X(S).
\end{align*}

The same final bound covers $|K_{s_p}|=1$. Indeed, when
$\diam(S)>0$, every closed tour through a positive-diameter set has length at
least twice its diameter, and hence $\mu_X(S)\geq2$. Combining
\eqref{eq:is-kappa} and~\eqref{eq:is-theta}, we obtain
\[
  \sum_k\beta_k^l r_{\mathrm{IB}}^j
  \leq96(t+1)\mu_X(S)\OPT.
\]
Equation~\eqref{eq:is-inner-pays-net} now gives
\[
  \sum_{r_i\in R^l}\phi_i
  <2\alpha
    \sum_k\beta_k^l r_{\mathrm{IB}}^j
  \leq192\alpha(t+1)\mu_X(S)\OPT.
\]
\end{proof}

\subsection{Completing the proof}

Summing Lemma~\ref{lem:is-one-outer-group} over the at most $L$ outer groups
gives, for every net-cost group $R_j$,
\begin{equation}
  \sum_{r_i\in R_j}\phi_i
  \leq
  192\alpha(t+1)L\mu_X(S)\OPT.
  \label{eq:is-one-net-group}
\end{equation}
Summing~\eqref{eq:is-one-net-group} over the $m$ net-cost groups and adding
\eqref{eq:is-easy-sum} yields
\begin{equation}
  \Phi
  \leq
  \Bigl(
    16\alpha+1
    +192\alpha(t+1)mL\mu_X(S)
  \Bigr)\OPT.
  \label{eq:is-total-net}
\end{equation}
Here
\begin{align*}
  m
  &=1+\left\lceil\log_2(tn)\right\rceil,\\
  L
  &=1+\left\lceil
       \log_2\!\left(
         \alpha n\left(t+1+\frac{t}{2\alpha}\right)
       \right)
     \right\rceil.
\end{align*}
For fixed $t$ and $\alpha$, both $m$ and $L$ are $O(\log(2n))$.

Applying Lemma~\ref{lem:is-online-from-net} to
\eqref{eq:is-total-net}, we obtain the explicit estimate
\begin{equation}
  \ALG=w(M_n)
  \leq
  \frac{2t}{\alpha(t-1)}
  \Bigl(
    16\alpha+1
    +192\alpha(t+1)mL\mu_X(S)
  \Bigr)\OPT.
  \label{eq:is-explicit-final}
\end{equation}
If $\diam(S)=0$, all server locations coincide and every feasible assignment
has the same cost, so the theorem is immediate. Otherwise
$\mu_X(S)\geq2$, and the additive term $16\alpha+1$ is absorbed by the
$\mu_X(S)mL$ term. Since $mL=O_{t,\alpha}(\log^2 n)$, we conclude that
\[
  \ALG
  \leq
  O_{t,\alpha}\!\left(\mu_X(S)\log^2 n\right)\OPT.
\]
This proves Theorem~\ref{thm:input-sensitive-restated}.
\qed

\section{Experiments}
\label{sec:experiments}
\suppressfloats[t]

\noindent\textbf{Setup.}
We compare RG with Greedy, RM,  the algorithm of \citet{arndt2026competitive}, and the offline optimum (OPT). Every method receives the same requests, sites, capacities, and arrival order. We report mean cumulative assignment cost, per-request runtime, and the fraction of decisions and cost attributable to nearest-available-server assignments for RG and, on NYC, Arndt et al.\ (ties count as nearest). Results are averaged over ten instances of $10{,}000$ requests. Our C++ implementation runs on one core of an Apple M2 processor with 8~GB RAM.

\noindent\textbf{Datasets.}
\textbf{NYC-Taxi} uses the samples of \citet{seth2026efficient} from the NYC Taxi data~\citep{KaggleNYCdata}: pickups are requests, drop-offs form the server-side point cloud, and costs are Euclidean distances in the original two-dimensional coordinates. Each sample contains $10{,}000$ trips from the first week of a fixed month. \textbf{MovieLens-25M}~\citep{HarperKonstan2015} uses the 1,128-dimensional Tag Genome vectors~\citep{VigSenRiedl2012} available for 13,816 movies. We standardize each tag coordinate and normalize each movie vector. A user's location is the normalized weighted sum of genome movies that the user rated at least four stars, with weight equal to the rating minus three. Restricting to users with at least 20 such ratings leaves 13,758 request vectors; the server-side points are the 10,000 genome movies with the most ratings. We fit PCA on these movies, map movies and users to 100 dimensions, and normalize them. For both datasets, $k$-means clusters the 10,000 server-side points into $k$ sites, with each site's capacity equal to its cluster size; clustering is rerun whenever $k$ changes.

\noindent\textbf{Arrival orders and parameters.}
For NYC, we use both the natural timestamp order and a \emph{directional} order obtained by sorting the same requests by their $y$-coordinate. The latter is a simple, non-adaptive stress test for spatially correlated arrivals: it changes only the order, while retaining every real request and server location. We set $k=500$, $t=6$, and $\alpha\in\{1,1.5,2\}$; we also vary $k\in\{50,100,200,500\}$ at $(t,\alpha)=(6,2)$. For MovieLens, we use a directional arrival order with $k=500$, $t=1.12$, and $\alpha\in\{1.02,1.05,1.08,1.10\}$. For computational efficiency, all online algorithms operate on the same fixed 12-dimensional Gaussian projection, while their resulting assignments and $\mathrm{OPT}$ are evaluated in the 100-dimensional reference space.

\begin{figure}[H]
\centering
\begin{minipage}[t]{0.495\linewidth}
\vspace{0pt}
\centering
{\scriptsize
\textbf{(a) NYC timestamp order}\\[-1pt]
$(k,t)=(500,6);\ \alpha=1,2$\par}
\includegraphics[width=.49\linewidth]{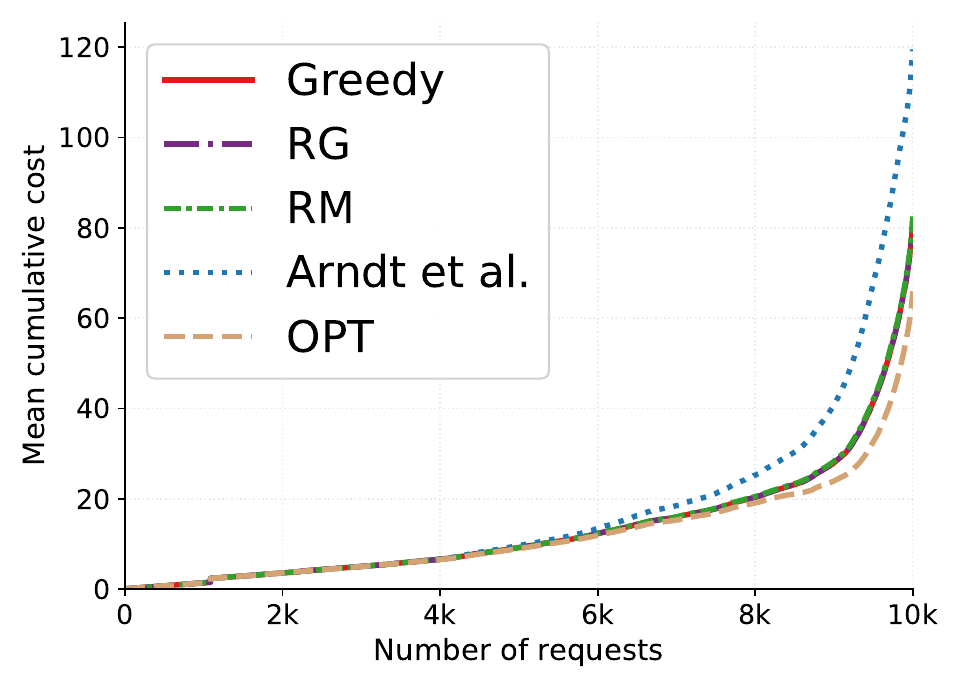}\hfill
\includegraphics[width=.49\linewidth]{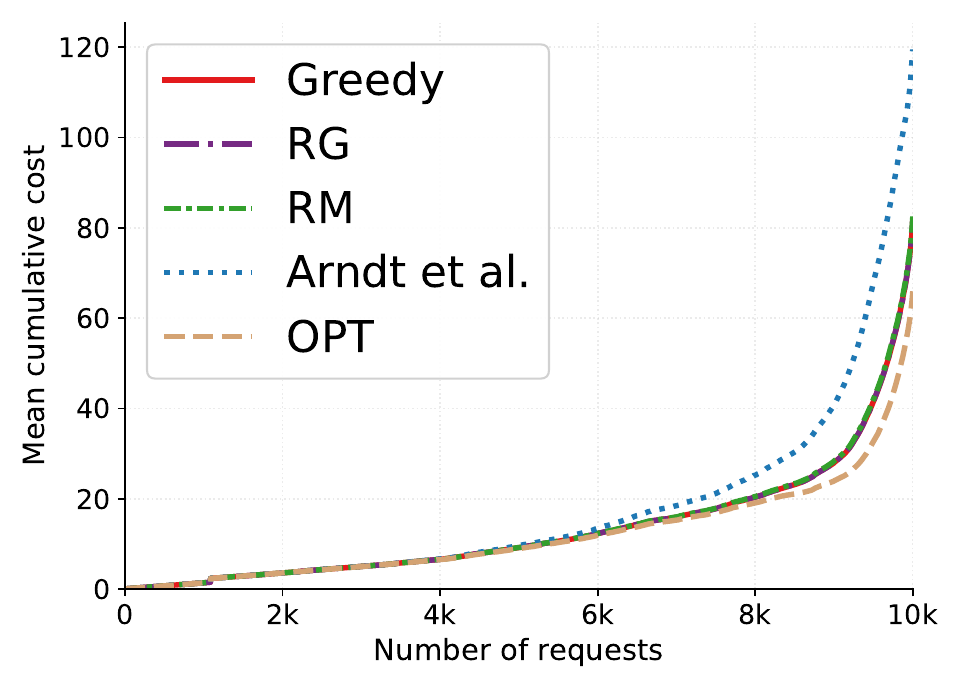}
\end{minipage}
\hfill
\begin{minipage}[t]{0.495\linewidth}
\vspace{0pt}
\centering
{\scriptsize
\textbf{(b) NYC directional order}\\[-1pt]
$(k,t)=(500,6);\ \alpha=1,2$\par}
\includegraphics[width=.49\linewidth]{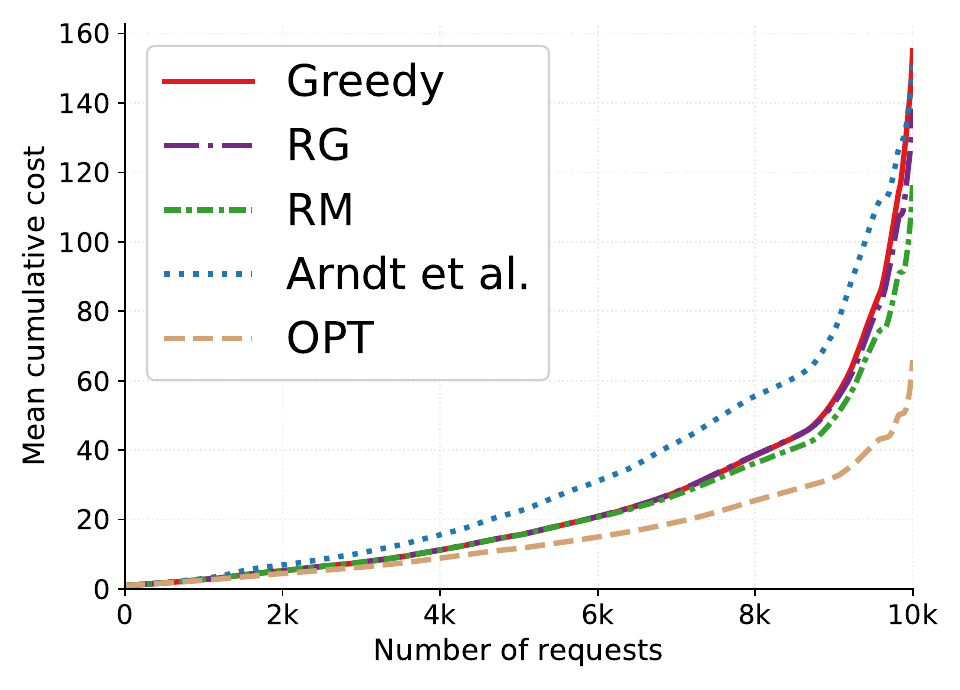}\hfill
\includegraphics[width=.49\linewidth]{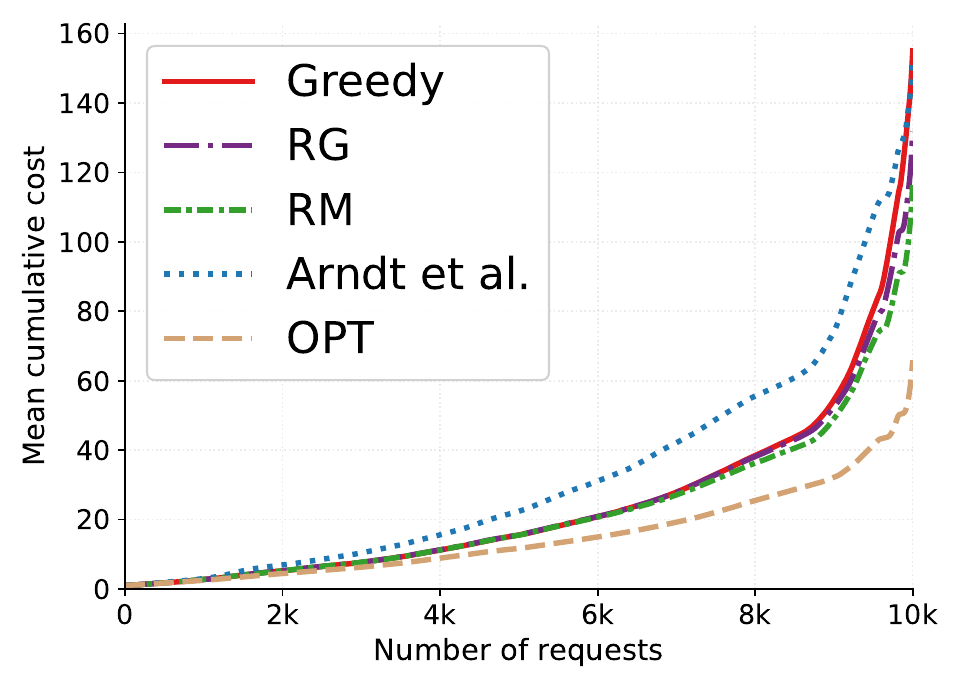}
\end{minipage}

\par\vspace{4pt}

\begin{minipage}[t]{0.495\linewidth}
\vspace{0pt}
\centering
{\scriptsize
\textbf{(c) NYC directional order: varying $k$}\\[-1pt]
$(t,\alpha)=(6,2);\ k=50,500$\par}
\includegraphics[width=.49\linewidth]{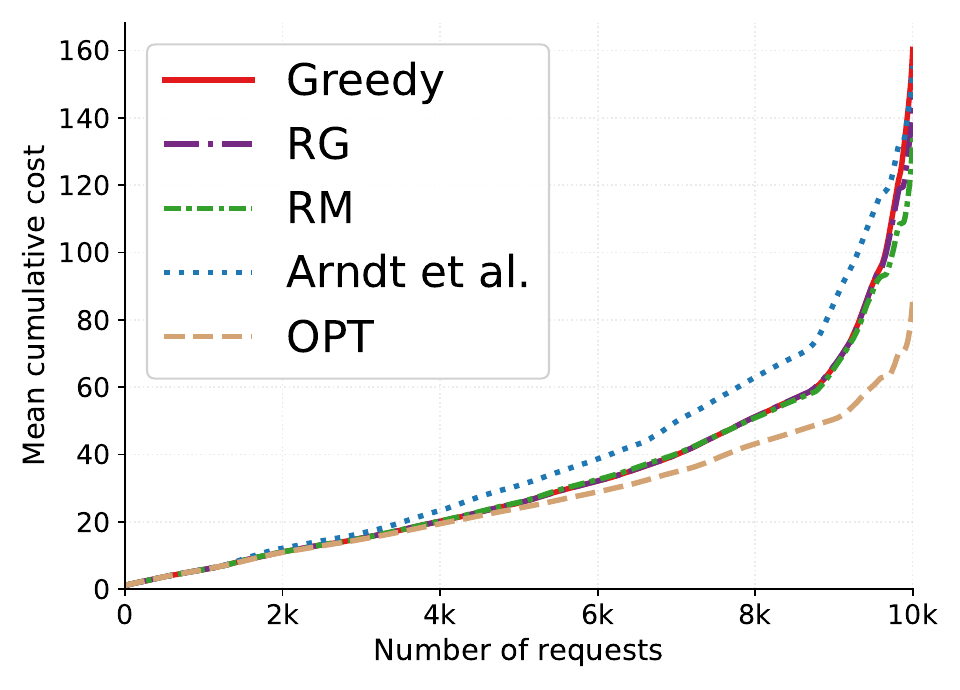}\hfill
\includegraphics[width=.49\linewidth]{Figures/OTR_Clustered_K500_DirectionalStochastic_T6_Alpha2_WithRM_Cost.pdf}
\end{minipage}
\hfill
\begin{minipage}[t]{0.495\linewidth}
\vspace{0pt}
\centering
{\scriptsize
\textbf{(d) MovieLens directional order}\\[-1pt]
$(k,t)=(500,1.12);\ \alpha=1.02,1.10$\par}
\includegraphics[width=.49\linewidth]{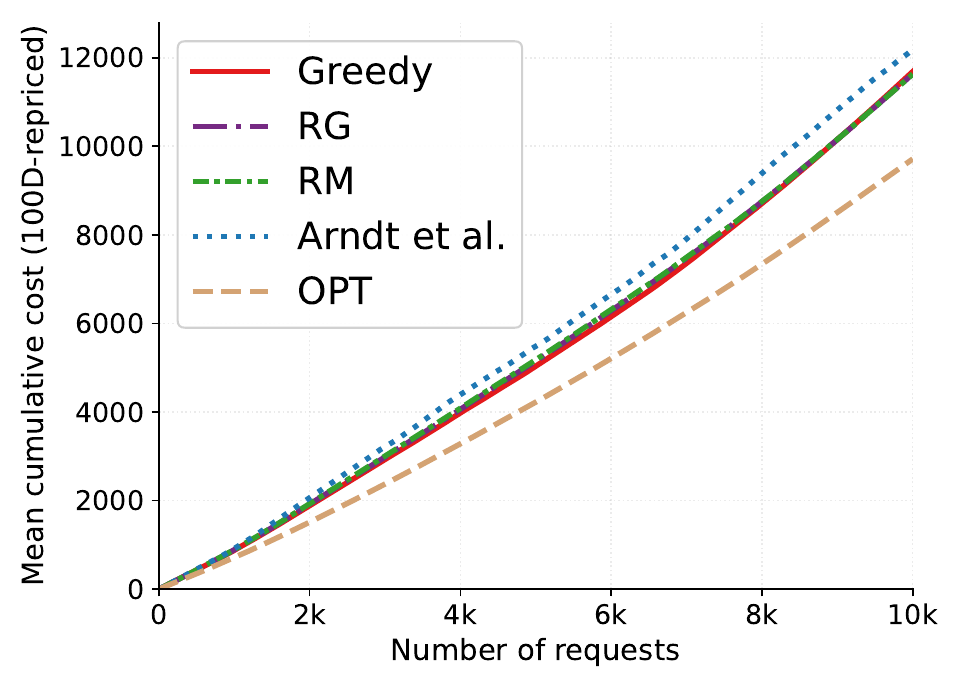}\hfill
\includegraphics[width=.49\linewidth]{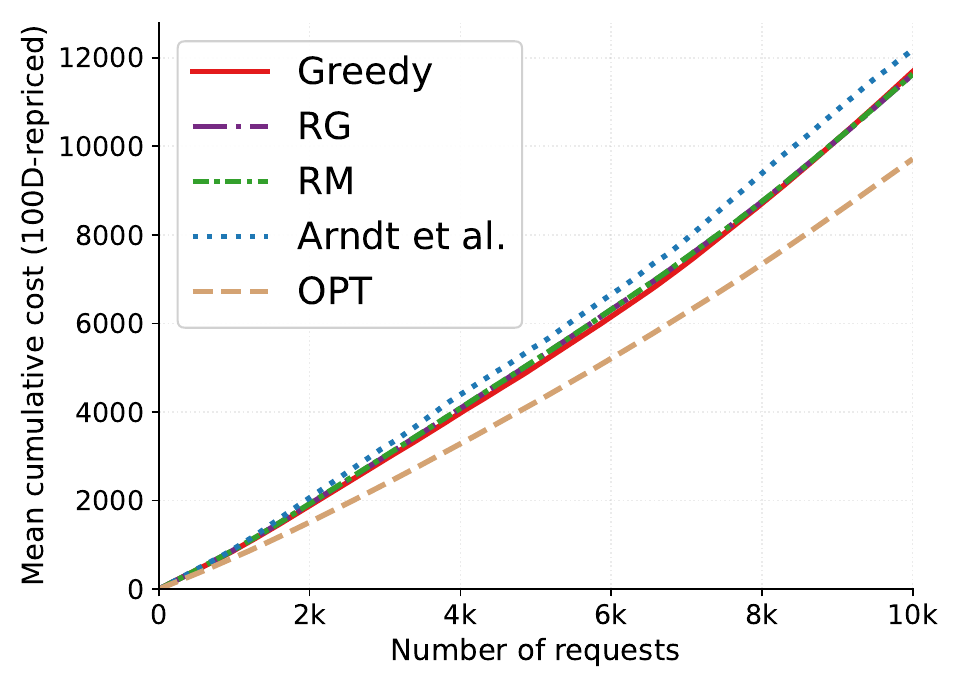}
\end{minipage}

\par\vspace{3pt}
\caption{Mean cumulative assignment cost over ten instances, shown at the endpoints of each parameter sweep. RG tracks Greedy on NYC's timestamp order, but is distinctly better under directional ordering; this advantage grows with $k$. On MovieLens, RG consistently improves on Greedy and both outperform Arndt et al. Complete sweeps appear in Appendix~\ref{app:additional-experiments}.}
\label{fig:experiments}
\end{figure}

\noindent\textbf{Assignment cost.}
On NYC's timestamp traces, RG stays within $0.5\%$ of Greedy for every tested $\alpha$, while both cost about $31\%$ less than Arndt et al. (Figure~\ref{fig:experiments}(a) and Appendix~\ref{app:additional-experiments}). RM costs $1.0\%$ more than Greedy; but still costs about $31\%$ less than Arndt et al. This is the expected behavior on natural traces, where Greedy is already known to be empirically strong~\citep{tong2016online}: RG preserves that favorable behavior while retaining its worst-case guarantee. The directional trace exposes the benefit of robustification: as $\alpha$ increases from $1$ to $2$, RG's reduction in final cost relative to Greedy grows from $10.75\%$ to $15.02\%$ (Figure~\ref{fig:experiments}(b)). The gain also strengthens with the number of sites, rising from $8.54\%$ at $k=50$ to $15.02\%$ at $k=500$ (Figure~\ref{fig:experiments}(c)). 
RM extends this trend to its limit; at $k=500$, $t=6$, it reduces cost by $24.3\%$ relative to Greedy.


On MovieLens, RG costs less than Greedy in all ten repetitions for every tested $\alpha, t=1.12$, with mean reductions between $0.41\%$ and $0.51\%$ across the sweep (Figure~\ref{fig:experiments}(d) and Appendix~\ref{app:additional-experiments}). Arndt et al. has higher cost than both methods in every run. RG therefore retains a consistent edge even in this high-dimensional setting where Greedy is already a strong baseline~\cite{tong2016online}. For MovieLens, we show the results for $t=1.12$ in the main text since it reveals a regime in which RG makes a substantial fraction of non-greedy assignments. At $t=6$, as used for NYC, RG matches Greedy's cost on MovieLens, while both outperform Arndt et al.; these results appear in Appendix~\ref{app:additional-experiments}.
Additional results for other tested value of $t, \alpha$ and $k$ on both the datasets are reported in Appendix~\ref{app:additional-experiments}.

\noindent\begin{minipage}{\linewidth}
\noindent\textbf{Runtime and nearest-neighbor behavior.}
\begin{center}
\begin{minipage}[t]{0.37\linewidth}
\centering
\footnotesize
\setlength{\tabcolsep}{3pt}
\renewcommand{\arraystretch}{1.05}
\textbf{Runtime per request (ms)}\\[2pt]
\begin{tabular*}{\linewidth}{@{\extracolsep{\fill}}lrr@{}}
\hline
Algorithm & NYC & MovieLens\\
\hline
RG & 0.04291 & 0.2647\\
RM & 0.06312 & 0.2674\\
Greedy & 0.00114 & 0.0018\\
Arndt et al. & 0.00152 & 0.0033\\
OPT & 0.5403 & 1.1660\\
\hline
\end{tabular*}
\end{minipage}
\hfill
\begin{minipage}[t]{0.60\linewidth}
\centering
\footnotesize
\setlength{\tabcolsep}{2pt}
\renewcommand{\arraystretch}{0.98}
\textbf{Nearest-neighbor contribution}\\[2pt]
\begin{tabular*}{\linewidth}{@{\extracolsep{\fill}}llcrr@{}}
\hline
Dataset & Algorithm & $\alpha$ & Decisions (\%) & Cost (\%)\\
\hline
NYC & RG ($t=6)$ & 1.00 & 98.05 & 87.78\\
    &    & 1.50 & 97.62 & 85.48\\
    &    & 2.00 & 97.37 & 84.05\\
\cline{2-5}
    & RM ($t=6$) & -- & 93.44 & 62.46\\
\cline{2-5}
    & Arndt et al. & -- & 51.34 & 9.06\\
\hline
MovieLens & RG ($t=1.12$) & 1.02 & 63.76 & 52.53\\
           &    & 1.10 & 51.79 & 40.79\\
\cline{2-5}
    & RM ($t=1.12$) & -- & 47.40 & 36.66\\
\cline{2-5}
    & Arndt et al. & -- & 14.10 & 7.17\\
\hline
\end{tabular*}
\end{minipage}
\end{center}
\end{minipage}
\par\smallskip

\noindent\textbf{Runtime trade-off.}
The runtime and nearest-neighbor entries are averaged over ten instances of $10{,}000$ requests; the runtime comparison uses $(t,\alpha)=(6,2)$ for NYC and $(1.12,1.10)$ for MovieLens. RG and RM perform more computation than the two online baselines, while both remain substantially faster than exact offline optimization. This is the computational trade-off for simultaneously obtaining a robust competitive guarantee and explicitly favoring nearest-neighbor assignments. The latter remain dominant: they constitute $97.37$--$98.05\%$ of RG's NYC decisions and $51.79$--$63.76\%$ of its MovieLens decisions. For comparison, Arndt et al.'s nearest-server assignments on NYC constitute only $51.34\%$ of its decisions and $9.06\%$ of its cost.

\noindent\textbf{Key takeaway.}
RG preserves Greedy's strong performance on natural traces, but is markedly more stable when arrivals have even mild, fixed spatial directionality: on unchanged NYC locations it improves over Greedy by up to $15.02\%$ in the regime $t=6, \alpha < (t-1)/2$, with a widening advantage as $k$ grows. Across both datasets it also consistently outperforms Arndt et al., while a majority of its decisions remain nearest-neighbor assignments.

\section{Conclusion}
We introduced Robustified Greedy, a deterministic algorithm for online transportation that interpolates between nearest-neighbor decisions and the robust augmenting-path corrections of Robust Matching. Its competitive ratio of $6.6604k-2.89$ improves the best previously known guarantee while remaining independent of the total capacity. The same framework also yields an $O(k)$ guarantee for Robust Matching and preserves the metric-sensitive $O(k^{1-1/d}\log^2 n)$ bound in fixed-dimensional Euclidean spaces.

The analysis identifies a structural reason for the improvement: positive-cost greedy edges persist until their servers saturate, allowing the cost to be accounted for over saturation epochs. This structure also provides an interpretable certificate for non-greedy decisions and guarantees that a substantial fraction of the total cost comes from nearest-neighbor assignments. The experiments are consistent with this picture: RG remains close to Greedy on natural arrival orders and is more robust under correlated arrivals. Natural directions for future work include narrowing the remaining gap to the $(2k-1)$ lower bound and determining whether the metric-sensitive logarithmic factors can be reduced.

\section*{AI use statement}

In this work, we used generative AI tools for  formulating mathematical claims, provide critical ingredients for proving mathematical claims and assisting in the writing of proofs, implementing methods and provide feedback on experiments.
We have reviewed all AI-assisted work.
We take responsibility for the final content of this work,
including text, claims or artifacts produced with the aid of generative AI.

\bibliographystyle{plainnat}
\bibliography{iclr2027_conference}

\clearpage
\appendix
\setcounter{section}{3}
\section{Additional Experimental Results}
\label{app:additional-experiments}

Figures~\ref{fig:experiments-full},~\ref{fig:experiments-full-NYC-k},~\ref{fig:experiments-full-movie-k1.12}, and~\ref{fig:experiments-full-movie-k6} report the complete parameter sweeps on
the NYC and MovieLens data.
All settings, instances, and evaluation procedures are identical to those in Section~\ref{sec:experiments}. 

\begin{figure}[!ht]
    \centering
    \setlength{\tabcolsep}{0pt}
    \renewcommand{\arraystretch}{1}
    \begin{tabular}{@{}ccc@{}}
        \multicolumn{3}{c}{\scriptsize\textbf{NYC timestamp order} ($k=500,t=6$; columns: $\alpha=1,1.5,2$)}\\[1pt]
        \includegraphics[width=.33\linewidth]{Figures/OTR_Clustered_K500_UniformStochastic_T6_Alpha1_WithRM_Cost.pdf} &
        \includegraphics[width=.33\linewidth]{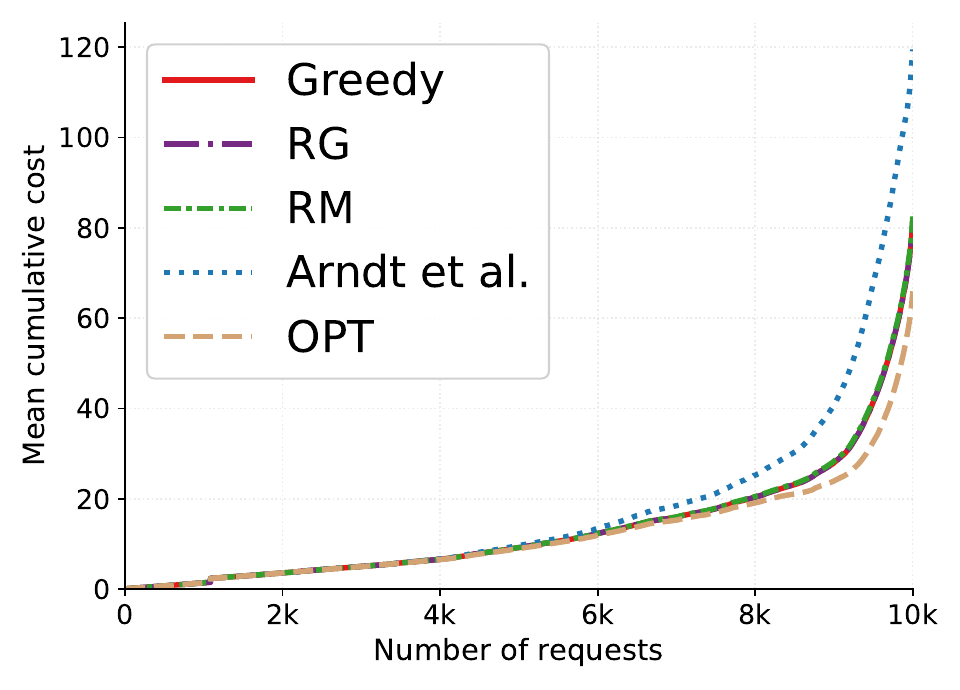} &
        \includegraphics[width=.33\linewidth]{Figures/OTR_Clustered_K500_UniformStochastic_T6_Alpha2_WithRM_Cost.pdf}\\[5pt]
        \multicolumn{3}{c}{\scriptsize\textbf{NYC directional order} ($k=500,t=6$; columns: $\alpha=1,1.5,2$)}\\[1pt]
        \includegraphics[width=.33\linewidth]{Figures/OTR_Clustered_K500_DirectionalStochastic_T6_Alpha1_WithRM_Cost.pdf} &
        \includegraphics[width=.33\linewidth]{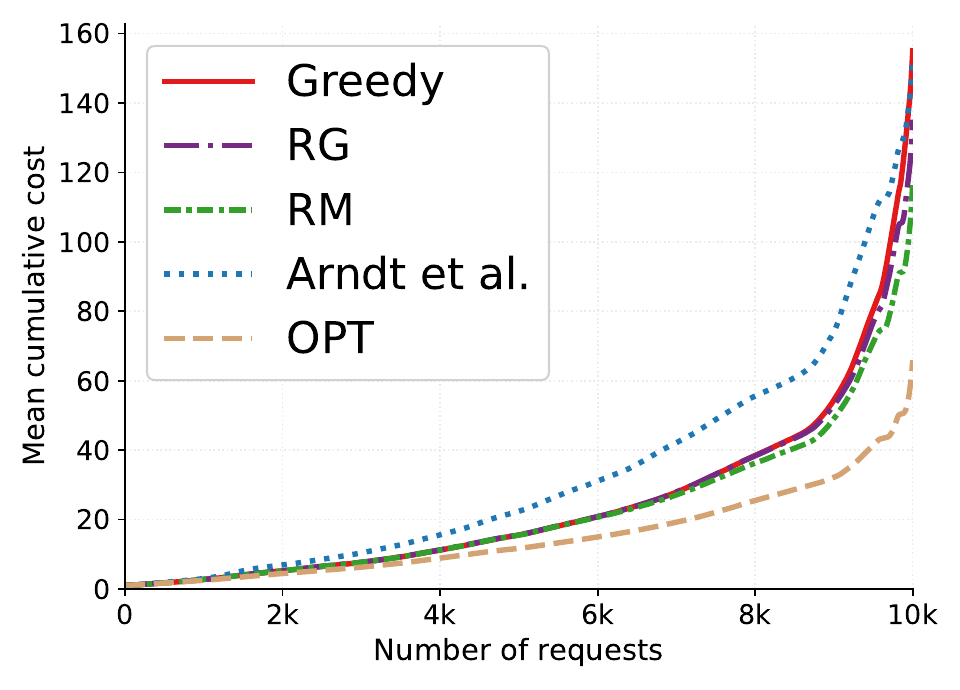} &
        \includegraphics[width=.33\linewidth]{Figures/OTR_Clustered_K500_DirectionalStochastic_T6_Alpha2_WithRM_Cost.pdf}\\[5pt]

    \end{tabular}
    \par\vspace{3pt}
    \caption{Complete mean cumulative assignment-cost sweeps over ten
    instances. The two rows vary $\alpha$ on NYC's timestamp and
    directional orders.}
    \label{fig:experiments-full}
\end{figure}

\begin{figure}[!ht]
    \centering
    \setlength{\tabcolsep}{0pt}
    \renewcommand{\arraystretch}{1}
    \begin{tabular}{@{}cccc@{}}
        \multicolumn{4}{c}{\scriptsize\textbf{NYC directional order, varying $k$} ($(t,\alpha)=(6,2)$; columns: $k=50,100,200,500$)}\\[1pt]
        \includegraphics[width=.245\linewidth]{Figures/OTR_Clustered_K50_DirectionalStochastic_T6_Alpha2_WithRM_Cost.pdf} &
        \includegraphics[width=.245\linewidth]{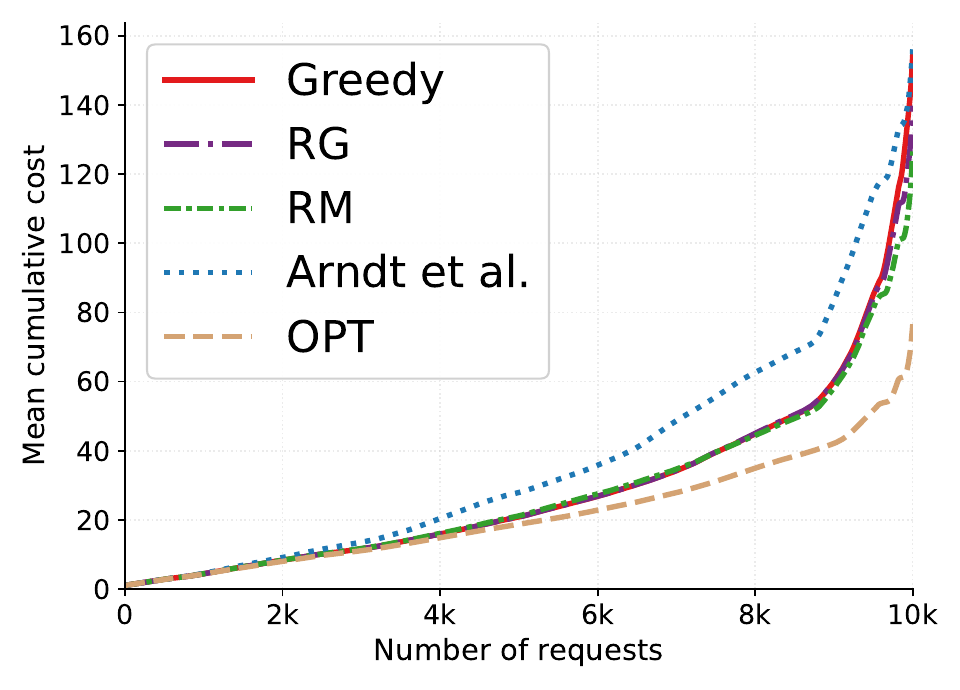} &
        \includegraphics[width=.245\linewidth]{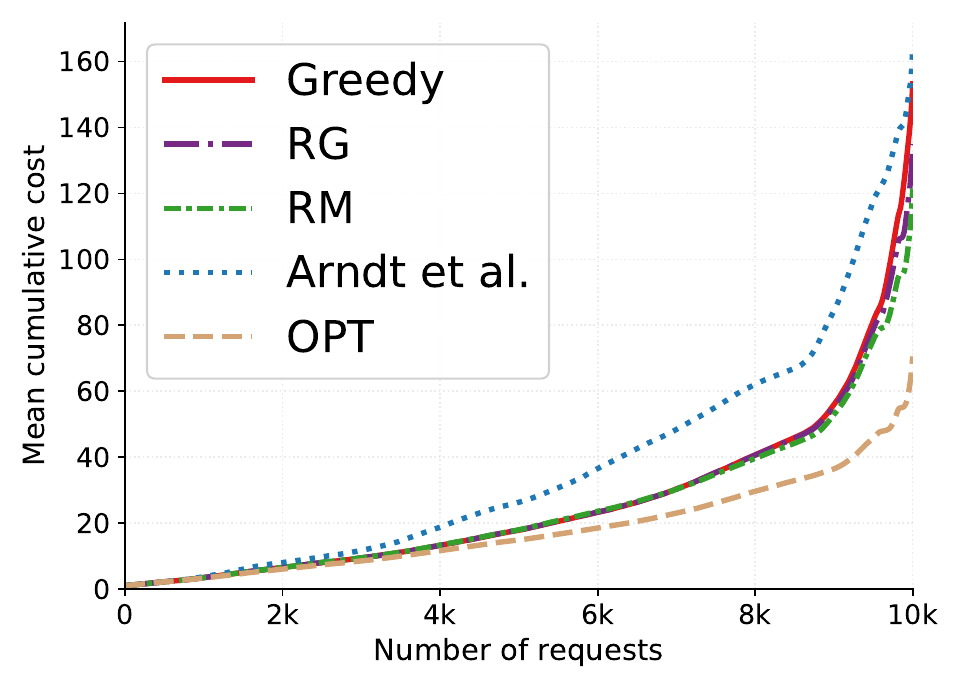} &
        \includegraphics[width=.245\linewidth]{Figures/OTR_Clustered_K500_DirectionalStochastic_T6_Alpha2_WithRM_Cost.pdf}\\[5pt]
    \end{tabular}
    \par\vspace{3pt}
    \caption{Complete mean cumulative assignment-cost sweeps over ten
    instances. Varies $k$ on NYC's directional order.}
    \label{fig:experiments-full-NYC-k}
\end{figure}

\begin{figure}[!ht]
    \centering
    \setlength{\tabcolsep}{0pt}
    \renewcommand{\arraystretch}{1}
    \begin{tabular}{@{}cccc@{}}
        \multicolumn{4}{c}{\scriptsize\textbf{MovieLens directional order} ($k=500,t=1.12$; columns: $\alpha=1.02,1.05,1.08,1.10$)}\\[1pt]
        \includegraphics[width=.245\linewidth]{Figures/MovieLens_K500_Directional_T1.12_Alpha1.02_WithRM_Cost_Repriced100D.pdf} &
        \includegraphics[width=.245\linewidth]{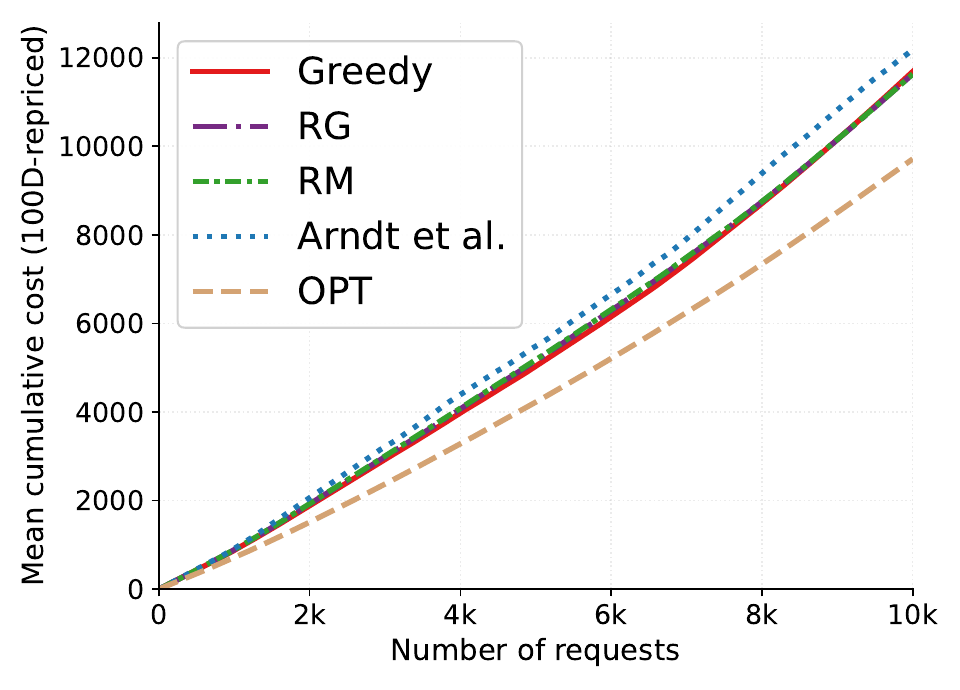} &
        \includegraphics[width=.245\linewidth]{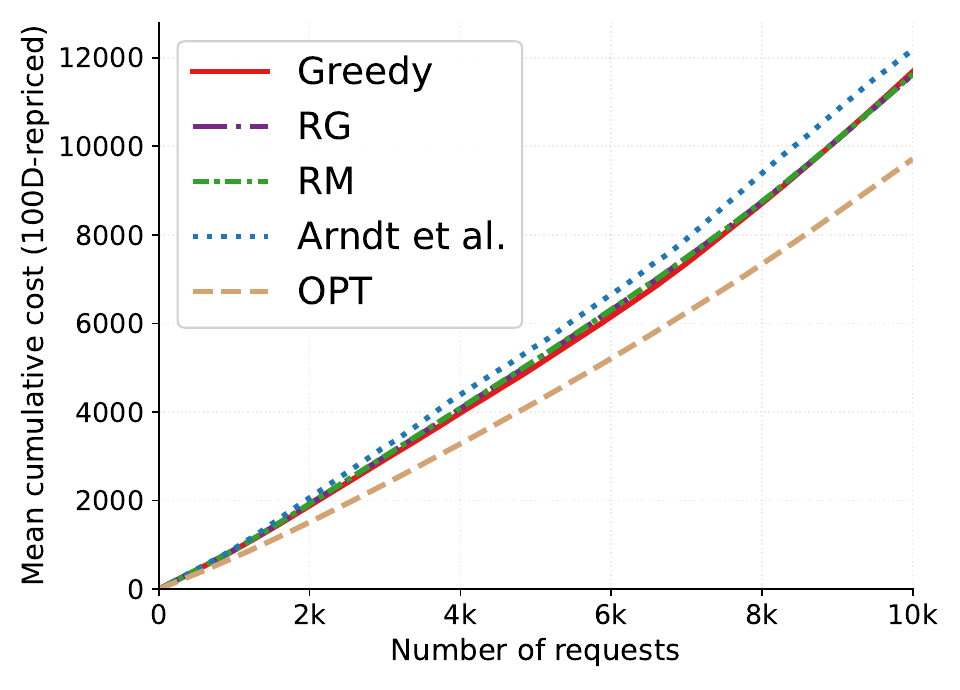} &
        \includegraphics[width=.245\linewidth]{Figures/MovieLens_K500_Directional_T1.12_Alpha1.10_WithRM_Cost_Repriced100D.pdf}\\
    \end{tabular}
    \par\vspace{3pt}
    \caption{Complete mean cumulative assignment-cost sweeps over ten
    instances on MovieLens for $t=1.12$ and $\alpha$.}
    \label{fig:experiments-full-movie-k1.12}
\end{figure}

\begin{figure}[!ht]
    \centering
    \setlength{\tabcolsep}{0pt}
    \renewcommand{\arraystretch}{1}
    \begin{tabular}{@{}ccc@{}}
        \multicolumn{3}{c}{\scriptsize\textbf{MovieLens directional order} ($k=500,t=6$; columns: $\alpha=1,1.5,2$)}\\[1pt]
        \includegraphics[width=.33\linewidth]{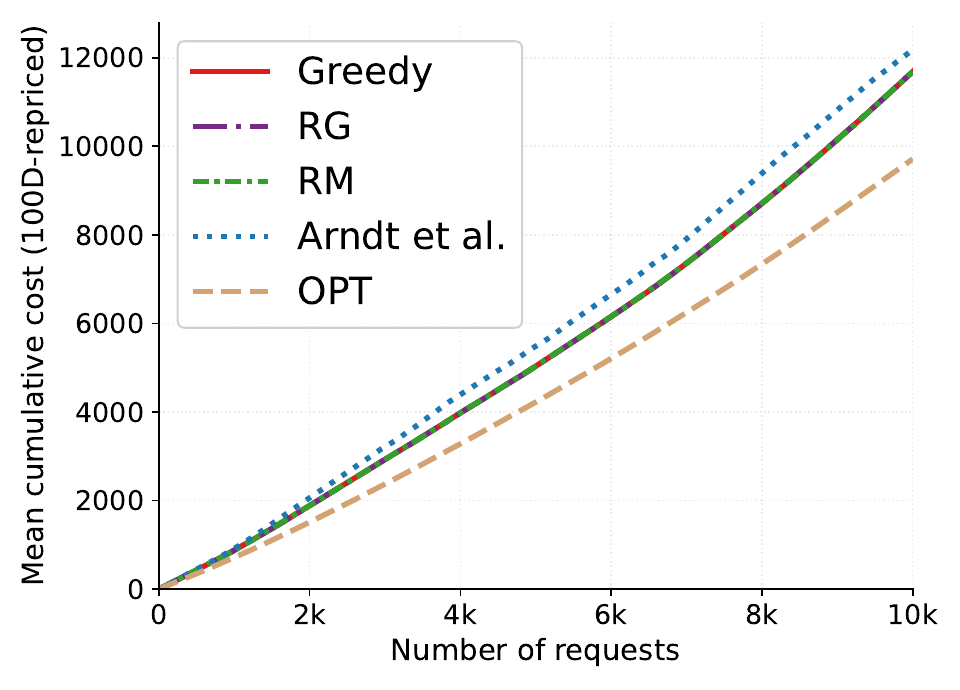} &
        \includegraphics[width=.33\linewidth]{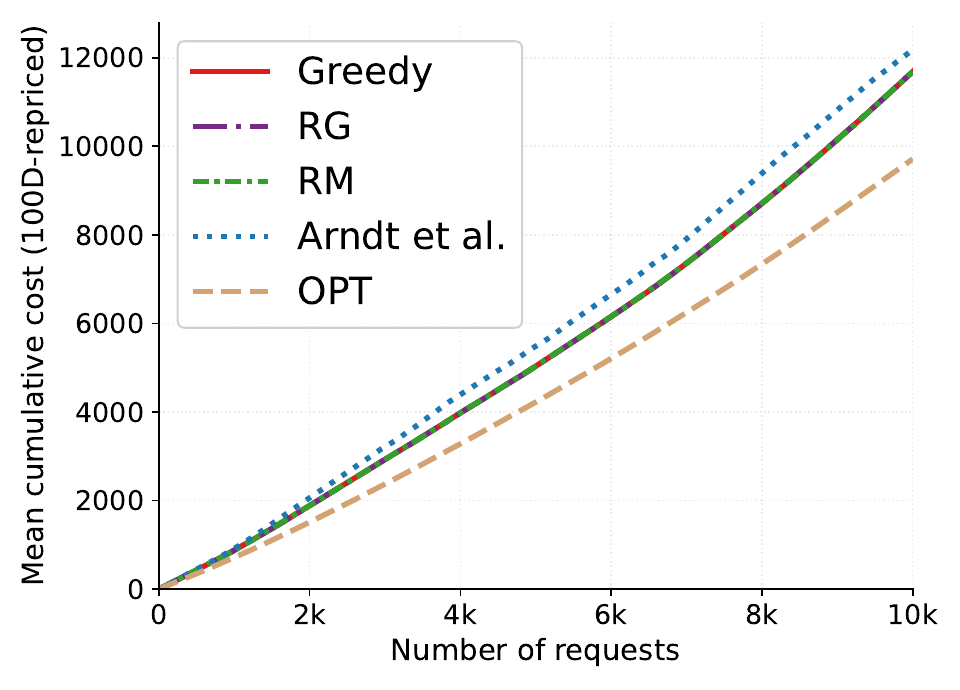} &
        \includegraphics[width=.33\linewidth]{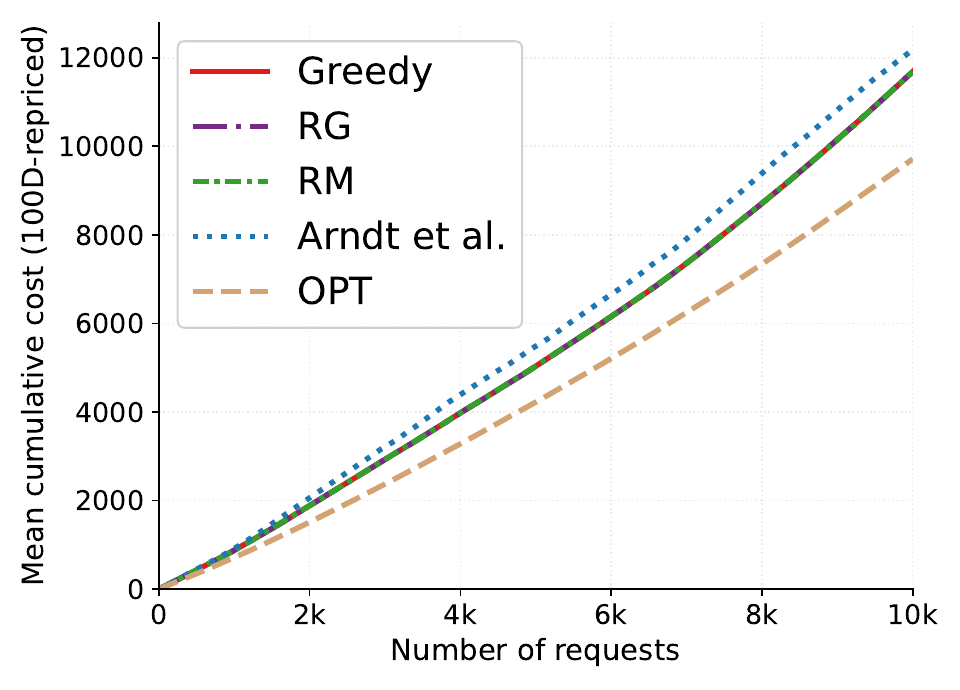}
    \end{tabular}
    \par\vspace{3pt}
    \caption{Complete mean cumulative assignment-cost sweeps over ten
    instances on MovieLens for $t=6$ and $\alpha$.}
    \label{fig:experiments-full-movie-k6}
\end{figure}

\end{document}